\documentclass[12pt,a4paper,fleqn]{scrartcl}

\usepackage{pdfpages}
\RequirePackage{amsthm,amsmath,natbib}
\RequirePackage{amssymb,amstext}
\RequirePackage{hypernat}
\usepackage{graphicx}
\usepackage{enumerate}
\usepackage{bbm}
\usepackage{url}
\usepackage{verbatim}
\usepackage{ bm}
\usepackage{setspace}

\DeclareOldFontCommand{\it}{\normalfont\itshape}{\mathit}

\usepackage[pdfpagemode=UseOutlines ,plainpages=false
,hypertexnames=false ,pdfpagelabels ,hyperindex=true,colorlinks=true]{hyperref}
\makeatletter%
	\Hy@breaklinkstrue%
	\makeatother%
\usepackage{color}
\definecolor{darkred}{rgb}{0.6,0.0,0.1}
\definecolor{darkgreen}{rgb}{0,0.5,0}
\definecolor{darkblue}{rgb}{0,0,0.5}
\hypersetup{colorlinks ,linkcolor=darkblue ,filecolor=darkgreen
,urlcolor=darkblue ,citecolor=darkblue}

\RequirePackage{natbib}
\numberwithin{equation}{section}
\newtheorem{theorem}{Theorem}[section]
\newtheorem{corollary}{Corollary}[section]
\newtheorem{example}{Example}[section]
\newtheorem{ass}{Assumption}
\newtheorem{proposition}{Proposition}[section]
\newtheorem{lemma}[proposition]{Lemma}
\newtheorem{remark}{Remark}[section]
\theoremstyle{definition}

\DeclareMathOperator{\E}{\mathbb{E}}

\def\Var{\mathop{{\mathbb V}ar}\nolimits}%

\newcommand{\R}{\mathbb{R}}

\renewcommand{\Pr}{\mathbb{P}}

\newcommand{\set}[1]{{\left\lbrace #1\right\rbrace }}
\DeclareMathOperator*{\argmin}{arg\,min}
\def\1{\mathop{\mathbbm 1}\nolimits}

\renewcommand{\cite}{\citet}
\begin{document}
\title{Varying Random Coefficient Models\thanks{
\, I would like to thank editor Elie Tamer and two anonymous referees for their comments.
I am also thankful to seminar participants at CREST Paris, Universit\"at G\"ottingen, the Berlin IRTG workshop,  the Workshop on Inverse Problems in Heidelberg, Universit\"at Bonn, the Econometric Study Group Conference in Bristol and Universit\"at Konstanz for their helpful suggestions.
Financial support by Deutsche Forschungsgemeinschaft through CRC TRR 190 is gratefully acknowledged.}}
\author{\textsc{Christoph Breunig}\thanks{Department of Economics, Emory University, Rich Memorial Building, Atlanta, GA 30322, USA, e-mail:
\url{	christoph.breunig@emory.edu}}\\
{\small \textit{Emory University}}
}

\date{\today }
\maketitle
\begin{abstract}
This paper analyzes unobserved heterogeneity 
when observed characteristics are modeled nonlinearly. The proposed model builds on {\it varying random coefficients} (VRC) that are determined by  nonlinear functions of observed regressors and additively separable unobservables. This paper proposes a novel  estimator of the VRC density based on weighted sieve minimum distance. The main example of sieve bases are Hermite functions which yield a numerically stable estimation procedure. This paper shows inference results that go beyond what has been shown in ordinary RC models.
 We provide in each case rates of convergence  and also establish pointwise limit theory of linear functionals, where a prominent example is the density of potential outcomes. In addition, a multiplier bootstrap procedure is proposed to construct uniform confidence bands. A Monte Carlo study examines finite sample properties of the estimator and shows that it performs well even when the regressors associated to RC are far from being heavy tailed. 
Finally, the methodology is applied to analyze heterogeneity in income elasticity of demand for housing.
\end{abstract}

\begin{tabbing}
\noindent \emph{Keywords:}  Random Coefficients, Varying Coefficients,  Sieve Minimum Distance,\\
Hermite  Functions, Rate of convergence, Bootstrap uniform confidence bands. 
\end{tabbing}
\begin{tabbing}
\noindent \emph{JEL classification:}  C14, C21\\[.2ex]
\end{tabbing}
\newpage
\section{Introduction}
Heterogeneity in individual behavior is a common source of variation in microeconometric applications.
Thus, in recent years it  became increasingly popular to explicitly model unobserved heterogeneity, for instance, by introducing random coefficients.
Yet identification in random coefficient models requires functional form restrictions. This paper shows that one can be much more flexible with respect to observed characteristics which can essentially influence the shape of the density of random coefficients. 
Specifically we extend the ordinary random coefficient model to allow for nonlinearities in observed heterogeneity, captured by varying coefficients.

The {\it varying random coefficient} (VRC) model is given by
\begin{align}
Y&=B_0+B_1'X,  \label{mod:gen}
\end{align}%
where $Y$ is a scalar dependent variable and $X$ is a vector of covariates of dimension $d-1$. 
The VRC vector $B=(B_0,B_1')'$ satisfies
\begin{align}
B_0=g_0(W)+A_0\quad \text{and}\quad 
B_{1,l}&= g_l(W)+A_{1,l}\quad\text{ where } 1\leq l\leq d-1\label{mod:gen:rc}
\end{align}%
for some covariates $W$ and unobservables $A=(A_0,A_{1,1}, \dots, A_{1,d-1})'$. The varying coefficient functions $g_0(\cdot),\dots,g_{d-1}(\cdot)$  are unknown and capture nonlinearities in observed heterogeneity.
 The vectors $X$ and $W$ may have elements in common but, to ensure identification of the varying coefficient functions in general, we rule out that $X$ is a subvector of $W$.
In this model, the {\it varying random slope} (VRS) given by $B_1$ represents observed and unobserved heterogeneity in the dependence of $Y$ on $X$. The VRC  is thus more general than the ordinary random coefficient (RC) model where all varying coefficient functions vanish and hence, $B=A$.  Identification of the VRC model is based on full independence of $A$ and $X$ but only conditional mean independence of $A$ and $(X,W)$.  While identification of the model requires $X$ to have enough variation, our setup permits $W$ to be discrete.

This paper is concerned with inference on the density of the VRC vector $B$ holding observed characteristics $W$ fixed. This density contains all the information of the underlying heterogeneity in the model and many functionals of it are of interest, e.g.,  the distribution of potential outcomes. 
A density estimator based on  weighted sieve minimum distance is proposed that builds on a conditional characteristic function equation of the model. 
The estimation criterion is minimized over a finite dimensional sieve space which is also convenient to impose shape restrictions on the estimator. The estimator is of closed form if no constraints are imposed on the sieve space and then coincides with a double series least squares estimator. An initial weighting step is used in the estimation criterion to stabilize the procedure. 
This is important, as it is well known that estimation of the joint RC density in  ordinary RC models leads to an ill-posed inverse problem. 
One insight of this paper is that our procedure allows us to separate estimation of the VRS density from estimation of the varying random intercept $B_0$. In particular, we show that estimation of the VRS density does not suffer from the ill-posed estimation problem once we impose finite dimensional restrictions on the density of $B_0$. 

For the sieve minimum distance estimator, inference results are established that go beyond what has been obtained in ordinary RC models. 
The rate of convergence of the estimator is derived, which coincides with the usual ill-posed rate of convergence when estimating the joint VRC density and corresponds to the usual well-posed, nonparametric  rate when only the density of  VRS is of interest and semiparametric restrictions on the random intercept are imposed.  Many important objects of interest, such as the distribution of the potential outcome, are  functionals of the density of $B$.  For a plug-in estimator of such linear functionals we establish pointwise asymptotic normality.
This paper also provides a bootstrap procedure to construct uniform confidence bands of the estimator. The inference results in this paper make explicit how the marginal distribution of $X$ affects the asymptotic behavior of the estimator.

Identification of the model requires a large  support condition of $X$ in general. Yet under mild assumptions, identification can be also achieved under bounded support via  extrapolation. 
This motivates the use of global basis functions for the sieve minimum distance estimator, such as the Hermite functions. When the sieve space is spanned by Hermite functions, we see that the estimator considerably simplifies as these basis functions form eigenfunctions of the Fourier transform. The estimator performs well in finite samples even when covariates $X$ are far from being  heavy tailed.  This is demonstrated in Monte Carlo simulations that also clarify how the variance of $X$ affects the mean integrated squared error of the estimator in finite samples. 

The estimation procedure is also applied to analyze heterogeneity in income elasticity of housing demand using German survey data. In our specification, estimated observed housing characteristics exhibit a nonlinear shape. The estimated density of heterogeneous income elasticity is unimodal with mode close to zero and  positively skewed. Uniform confidence bands allow to make significant statements about the shape of the estimated density. The empirical application demonstrates that our proposed methodology can be useful to analyze complex heterogeneity using cross sectional data.

Nonparametric identification and estimation of ordinary RC models is considered by \cite{beran1992}, \cite{beran1996}, and \cite{hoderlein2010}. For testing of qualitative features of the ordinary RC models see \cite{dunker2017}.
\cite{lewbel2017unobserved} generalize these models to allow for nonlinear index functions and in the next section we will provide a more detailed comparison of it to the VRC model. 
The ordinary RC models can be extended to {\it conditional random coefficient models} that assume model equation \eqref{mod:gen} together with the condition that $X$ is independent of $B$ conditional on $W$. This model is more flexible than $(\ref{mod:gen}$--$\ref{mod:gen:rc})$, but requires that the conditional density of $X$ given $W$ satisfies a large support condition for all realizations of $W$. Beyond this more restrictive support restriction,  estimation in conditional RC models also  suffers from the curse of dimensionality of $W$. This is why functional form restrictions are typically employed rather than considering the more general conditional RC models, see for instance, \cite[p. 1120]{lewbel2017unobserved}.
Recently,  {\it correlated random coefficient models} were studied in the literature which allow for full dependence of random coefficients and covariates $X$. 
In this setting, instruments are available that drive the covariates $X$ but not the random coefficients in \eqref{mod:gen}. These types of models are analyzed by
\cite{masten2014} and \cite{hoderlein2014}. 
While such a model is clearly more general than the VRC model, identification of the CRC models can be challenging with more restrictive exclusion assumptions and large support conditions on the instruments.

The methodology of sieve estimation became increasingly popular in recent years. For sieve estimation of  conditional moment restrictions models see  \cite{NP03econometrica} and \cite{AC03econometrica}. The VRC model does not fall into this category.
 For sieve estimation of ordinary RC models with discrete outcome see \cite{fox2011simple} and \cite{fox2016simple}. In binary choice models,  \cite{gautier2011} proposed an estimator  based on needlet thresholding. In the context of specification testing, a sieve approach was used by \cite{breunig2016}. In the literature on varying coefficients, series estimators were analyzed by  \cite{xia1999}, \cite{fan2003}, \cite{xue2012}, or \cite{ma2015}. 

The remainder of the paper is organized as follows. Section \ref{sec:model} provides the setup, motivating examples, and sufficient conditions for nonparametric identification. In Section \ref{sec_inference}, the estimation procedure based on sieve minimum distance is introduced and its asymptotic properties are established. Section \ref{sec_MC} is concerned with  the finite sample properties of the estimator analyzed via Monte Carlo simulation and an empirical illustration. Section \ref{sec_Conclusion} concludes. All proofs can be found in the appendix.

\section{The Model and Identification}\label{sec:model}
This section consists of two subsections. Subsection \ref{sub_sec:vrc} recalls the  varying random coefficients (VRC) model, outlines its key properties, and provides motivating examples for it. 
Subsection \ref{sub_sec:ident} provides an identification result of the joint density of the VRC vector $B$. 
\subsection{The Varying Random Coefficient Model}\label{sub_sec:vrc}
Consider again equations $(\ref{mod:gen}$--$\ref{mod:gen:rc})$, the VRC model is given by
\begin{align}
Y&=B_0+B_1'X, \tag{\ref{mod:gen}}\\
B_0&=g_0(W)+A_0\quad \text{and}\quad B_{1,l}= g_l(W)+A_{1,l}\quad\text{ where } 1\leq l\leq d-1.\tag{\ref{mod:gen:rc}}
\end{align}
As stated above $X$ and $W$ may have elements in common. Yet without further functional form restrictions on the varying coefficient functions $g_0,\dots,g_{d-1}$ we need to rule out that $X$ has only elements that are contained in $W$ (see also Assumption \ref{ass:ind} and the discussion thereafter). 
The covariates $X$ are assumed to be independent of $A$ and  the vector of covariates $(X',W')'$ is restricted to be mean independent of $A$, i.e., $\E[A|X, W]=0$ (see Assumption \ref{ass:ind} below). All the results in this paper will hold if there were no varying coefficients, i.e., model $(\ref{mod:gen}$--$\ref{mod:gen:rc})$ is the ordinary RC model. While identification of the model requires $X$ to have enough variation, our setup permits $W$ to be discrete.

Under conditional mean independence of $A$ and $(X,W)$,  model $(\ref{mod:gen}$--$\ref{mod:gen:rc})$ implies the varying  coefficient model
\begin{align}\label{mod:vcim}
\E[Y|X,W]= g_0(W)+\sum_{l=1}^{d-1} g_l(W)X_l=:g(S),
\end{align}
using the notation $S=(X',W')'$.
In the conditional mean restriction \eqref{mod:vcim}, the varying coefficient functions $g_l$ are identified through a rank condition, see also below. For an overview article of varying coefficient models see \cite{park2015varying}.
We emphasize that the varying coefficients specification in \eqref{mod:vcim} also derives from the additive separability of the random coefficients $A$ in equation \eqref{mod:gen:rc}. Without imposing such an additively separable structure, our model is in general not identified, see \cite[Corollary 2]{masten2014}.
On the other hand, under further assumptions, \cite{lewbel2017unobserved} establish identification when equation \eqref{mod:gen} is replaced by $Y=B_0+\sum_{l=1}^{d-1}G_l(B_{1,l}X_l)$ for unknown functions $G_l$  and $A$ is independent of $(X,W)$ (it is thus non-nested to our VRC model). While the sieve minimum distance approach in this paper allows for estimation of additional nonlinear index functions, we do not consider this extension. 

Identification of the varying coefficient functions $g_0,\dots,g_{d-1}$ in the VRC model $(\ref{mod:gen}$--$\ref{mod:gen:rc})$ can be obtained in the following two scenarios: First,  under a full rank condition on $\E[XX'|W=w]$  or, second, if $g_1\equiv\dots\equiv g_{d-1}\equiv0$. In the first case, the VRC model has the interpretation similar to a conditional random coefficient model but, instead of leaving the distribution of $X$ and $W$ unrestricted,  it imposes more structure on it. 
Also similarly to conditional random coefficient models, the role of $W$ is to ensure that independence between $A$ and $X$ is more plausible. The variables $W$ can also serve as control function residuals, which then allows $X$ to be endogenous, i.e., $A$ is unconditionally correlated with $X$.\footnote{Following \cite{imbens2009} assume that endogenous regressors $X$ are related to observed instruments $Z$ via $X=\psi(Z,W)$ where $\psi$ is strictly monotonic in scalar $W$ and $Z\perp(B,W)$. The model implies independence of $X$ and  $A$ conditional on $W=F_{X|Z}(X|Z)$.}
In the second case, only $g_0$ differs from the zero function and thus the VRC model has the interpretation of a nonlinear model in $X$, when $X=W$,  but with random coefficients only on the linear term. 

This paper is concerned with estimation of the VRC density holding the observed characteristics fixed at some potential realization $w$, i.e., the density $f_B(\cdot,w)$ of 
\begin{align*}
B^w= \mathbf g(w)+A
\end{align*}
where $\mathbf g(\cdot)=(g_0(\cdot),\dots,g_{d-1}(\cdot))'$ denotes the vector of varying coefficient functions.
In the case where $X$ and $W$ have no common elements, $B^w$ captures heterogeneous marginal effects. If $X$ and $W$ have joint elements then, to obtain marginal effects, replace $\mathbf g(w)$ with the vector of partial derivatives of $\mathbf g(w)$ w.r.t. $x$. The results in this paper are still valid in this case but it is not made explicit in order to keep the notation simple. 
By holding observed characteristics $W$ fixed, the density $f_B(\cdot,w)$ contains all information on unobserved heterogeneity.
Many objects of interest, such as the density of potential outcomes of $Y$, are linear functionals of  $f_B(\cdot,w)$.

Economic theory and empirical findings suggest nonlinearities in many applications of interest. For instance, through a nonparametric analysis of Engel curves to analyze nonlinearities in total expenditure, \cite{banks1997} suggest  quadratic terms in the logarithm of total expenditure.
While a random coefficient version of their nonlinear model is not identified one might still account for unobserved heterogeneity by allowing only the coefficient for the linear term to vary among individuals. The following two examples provide a relation of VRC to measurement error models and show that ignoring nonlinearities in the varying coefficients may have severe consequences in a standard Monte Carlo exercise setting.

\begin{example}[Measurement Error Models]
Consider a regression model with interaction term and measurement error in one covariate:
\begin{align}\label{mod:me}
Y= \beta_0+\beta_1 X_1 + \beta_2 X_1 X_2^*+ \beta_3 X_2^* + U,
\end{align}
where the variable $X_2^*$ is observed only with measurement error, i.e., $X_2=X_2^*+V$.
The deterministic parameters $\beta_l$, $1\leq l\leq 3$, are unknown, and  $(U,V)$ are unobservables with zero mean. Assume that an additional variable $W$, the instrumental variable, is available which satisfies $X_2^*=m(W)+\widetilde A$ where  $W$ is conditional mean independent of $\widetilde A$ and $V$ (see, for instance, \cite{hausman1991}, \cite{schennach2007instrumental}, or \cite{ben2017identification}). The conditional mean restriction identifies the parameters $\beta_0,\dots, \beta_3$ since the function $m(W)=\E[X_2|W]$ is identified. 
Using the instrumental variables approach we can rewrite the model \eqref{mod:me} as
\begin{align*}
Y= \underbrace{\beta_0+\beta_3m(W)}_{g_0(W)}+\underbrace{\beta_3\widetilde A+U}_{A_0}+
\big(\underbrace{\beta_1+ \beta_2m(W)}_{g_1(W)}+\underbrace{\beta_2\widetilde A}_{A_1}\big)    X_1
\end{align*}
and thus, is a special case of the VRC model $(\ref{mod:gen}$--$\ref{mod:gen:rc})$.
\end{example}

\begin{example}[Monte Carlo Simulation under misspecification]\label{example:MC}
 This finite sample example shows that falsely assuming linearity of varying coefficient functions may lead to severe biases that go beyond  typical approximation errors. 
Here, draws of regressors $(X,W)$ are generated from the bivariate standard normal distribution.  Let $Y=A_0+B_1X$ where $B_1=2W^2+A_1$ and  random coefficients $(A_0,A_1)$ are generated independently of $(X,W)$ as follows: $A_1$ is drawn from a mixture of normal distributions, i.e., $\mathcal{N}\left( -1.5,2\right) $ and $\mathcal{N}\left( 1.5,1\right) $ with weights $0.5$, and independently from $A_0\sim\mathcal{N}( 0,1)$. 
\begin{figure}[ht]
	\centering
		\includegraphics[width=15cm]{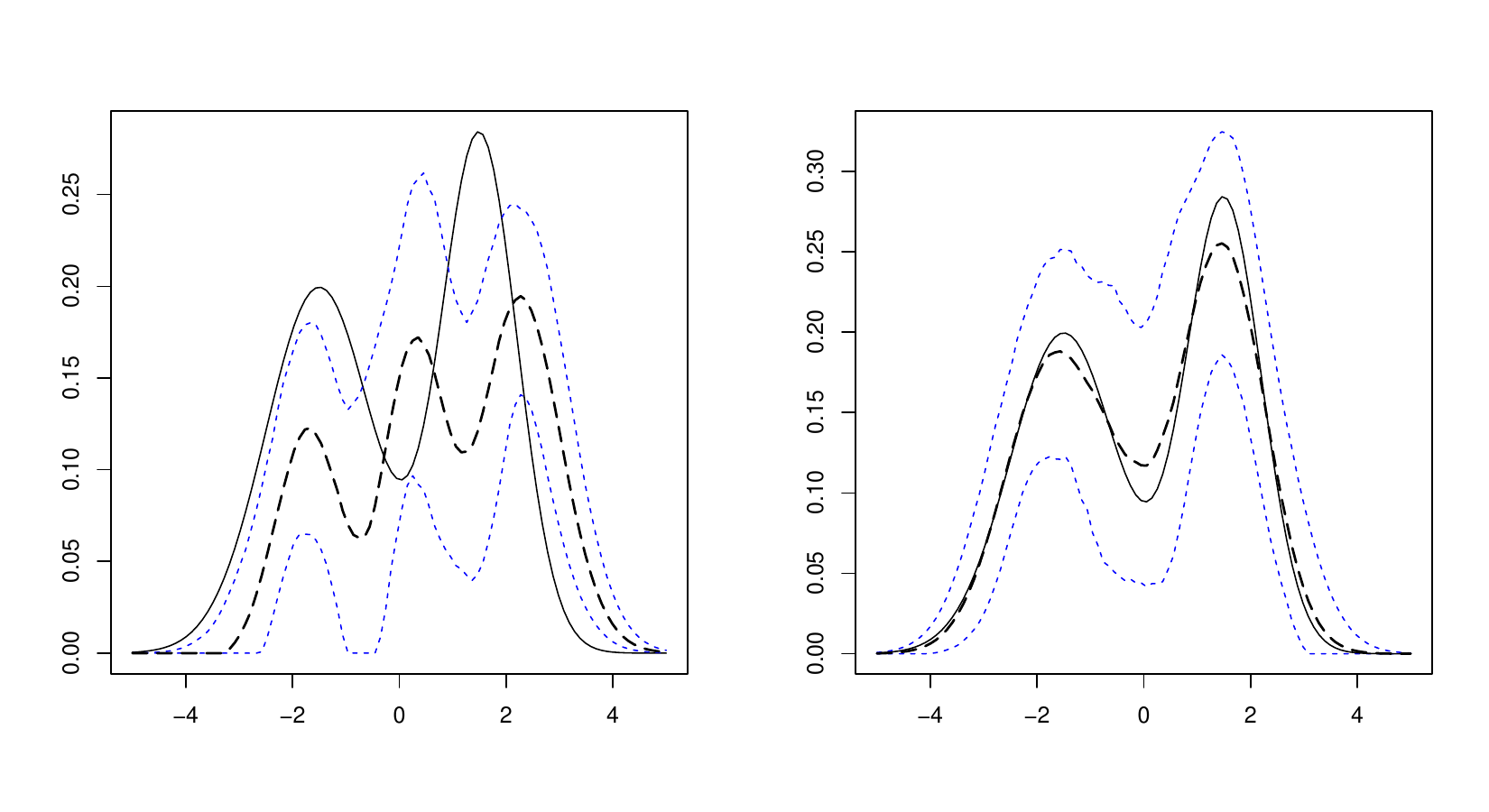}
		\vskip -1cm
	 \caption{{\small Solid lines: true density $f_{B_1}(\cdot,w)$ with $w=0$. Dotted lines: median and  pointwise $95\%$ confidence intervals. Left: Using OLS estimator of $g_1$. Right: Using B-splines estimator of $g_1$. Sample Size: $n=1000$. Monte Carlo iterations: $1000$. }}\label{misspec}
\end{figure}

We implement our estimator using Hermite functions, as described in the Monte Carlo section, but estimate the function $g_1(w)=2w^2$ once via ordinary least squares (OLS) and once via quadratic B-splines with three interior knots. Figure \ref{misspec} shows on the left the resulting estimator when $g_1$ is estimated via OLS and on the right when $g_1$ is estimated via B-splines.  
From Figure \ref{misspec} we see that ignoring the nonlinearity in $g_1$ can imply additional nonlinearities in the resulting estimator of the density of $B_1$.
\end{example}

\subsection{Identification}\label{sub_sec:ident}
This section provides assumptions under which the density $f_B(\cdot,w)$ of $B^w$ is identified for any value $w$ in the support of $W$. We now impose restrictions on the observed and unobserved variables of our model. 

\begin{ass}
\label{ass:ind} (i) $X$ is independent of $A$. (ii) $(X',W')$ is mean independent of $A$, i.e.,  $\E[A|X,W]=0$. (iii) $\mathbf g(\cdot)=(g_0(\cdot),\dots,g_{d-1}(\cdot))'$ is identified through the conditional moment equation \eqref{mod:vcim}.
\end{ass}

An independence assumption similar to Assumption \ref{ass:ind} is common in the literature on RC models with cross-sectional data (see, for instance, \cite{beran1993}, \cite{beran1996}, \cite{hoderlein2010}). It should be also emphasized that the independence assumption can often be justified in our model if the information in $W$ about heterogeneity is rich enough and in this sense is milder than in the ordinary RC model where $\mathbf g\equiv 0$. Clearly, if $W$ contains only elements in $X$ and $\E[A]=0$ then Assumption \ref{ass:ind} (ii) is implied by (i). 

Assumption \ref{ass:ind} (iii) is automatically satisfied if $g_l\equiv 0$ for all $1\leq l\leq d-1$.  Otherwise, Assumption \ref{ass:ind} (iii) is satisfied if for all $w$ in the support of $W$, the smallest
eigenvalue of $\E[XX'|W=w]$ is bounded away from zero. This rank condition is commonly imposed in the varying coefficient literature. 
In particular, it rules out that the vector of regressors $X$ has only values that are also contained in the vector $W$. It is also possible to relax the rank condition by imposing functional form restrictions on $g$, see \cite{fan2003}. 
Throughout the paper, the conditional characteristic function of $Y-g(S)$ given $X$ is denoted by $h(x,t;g)=\E\big[\exp\big(it(Y-g(S))\big)\big|X=x\big]$.
\begin{ass}\label{large:supp}
(i) $X$ has full support $\mathbb R^{d-1}$.
(ii) $\int_{\mathbb R^{d-1}} \int_{\mathbb R}  |t|^{d-1} |h(x,t;g)|dt\, dx<\infty$.
\end{ass}
While large support conditions are often required in econometrics to ensure identification, Assumption \ref{large:supp} (i) can be relaxed. If the distribution of $A$ has finite absolute moments and is uniquely determined by its moments, then identification with bounded support of $X$ can be achieved by extrapolation, see \cite{masten2014} and \cite{hoderlein2014}. Assumption \ref{large:supp} (ii) imposes a mild regularity assumption on the conditional characteristic function $h$. 

The next result establishes identification of the VRC density. We make use of identification of the varying coefficients through the conditional mean restriction \eqref{mod:vcim}. Consequently, by employing the relation $f_{ B}(b,w)= f_A(b-\mathbf g(w))$ the following result is due to Fourier inversion.  
\begin{lemma}\label{thm:ident:me}
Let Assumptions \ref{ass:ind}--\ref{large:supp} be satisfied. Then, for all $w$ in the support of $W$, the density of $B^w$ satisfies
\begin{align*}
f_{B}(b,w)=\frac{1}{(2\pi)^{d}} \int_{\mathbb R^{d-1}}\int_{\mathbb R} |t|^{d-1}\exp\big(-i t(1,x')(b- \mathbf g(w))\big) h(x,t;g) dt\,dx.
\end{align*}
\end{lemma}
Lemma \ref{thm:ident:me} shows that the density $f_B(\cdot,w)$ can be written as a transform of varying coefficient functions. Besides the shift of the conditional characteristic function $\E[\exp(itY)|X=x]$ to $h(x,t;g)$ there is also a location shift by $\mathbf g(w)$. This corresponds to shifts in frequency and time domain for the Fourier transform. 

\section{Estimation and Inference}\label{sec_inference}
This section presents an estimator of the VRC density based on sieve minimum distance and establishes its asymptotic properties. Subsection \ref{sub:sec:est} introduces the weighted sieve minimum distance estimator and motivates the use of Hermite functions as sieve basis. In Subsection \ref{sub:sec:rate}, the rate of convergence of the estimator of $f_B$ is derived. 
Subsection \ref{sub:sec:pw:inference} establishes the pointwise asymptotic distribution of linear functionals of $f_B$, where the density of potential outcome is one particular example of interest.  Subsection \ref{sec:bootstrap:conf} presents a Bootstrap procedure to construct uniform confidence bands. 

\subsection{The Sieve Minimum Distance Estimator}\label{sub:sec:est}
Estimation builds on the conditional characteristic function equation induced by the model $(\ref{mod:gen}$--$\ref{mod:gen:rc})$.
We denote the Fourier transform by $(\mathcal{F}\phi )(u):= \int_{\mathbb R^d} \exp
(iu'z)\phi (z)dz$ for any absolutely integrable function $\phi$.
Recall the notation of the conditional characteristic function $h(x,t;g)=\E[\exp(it(Y-g(S)))|X=x]$. 
Independence of $X$ and $A$, as imposed in Assumption \ref{ass:ind}, immediately implies
\begin{align*}
\E\big[\exp\big(it(A_0+A_1'X)\big)\big|X=x\big]=\int_{\mathbb R^d} \exp\big(it(1,x')a\big) f_A (a)da=\big(\mathcal F f_A\big)(t,tx)
\end{align*}
and hence the relation
\begin{align}\label{est:key:cond}
\big(\mathcal F f_A\big)(t,tx)
=h(x,t;g)
\end{align}
for all $t\in\mathbb R$ and $x\in\R^{d-1}$.
Moreover, relation \eqref{est:key:cond} leads to
\begin{align}\label{est:key:cond:L2}
\int_{\R^{d-1}}\int_\R\Big|\big(\mathcal F f_A\big)(t,tx)-h(x,t;g)\Big|^2 d\nu(t)dx=0,
\end{align}
for some measure $\nu$. Using this $L^2$ criterion we construct a sieve minimum distance estimator of $f_A$ and use a plug-in approach to estimate the VRC density $f_B$. Below, we show that the choice of the log-normal distribution as weighting measure $\nu$ is well suited for our estimation problem. 

The proposed sieve minimum distance estimator of the VRC density $f_{B}(\cdot,w)$ is based on the relation $f_{B}(b,w)= f_A(b-\mathbf g(w))$ for any $w$ in the support of $W$:
Consider the plug-in estimator
\begin{align}\label{est:denst:A}
\widehat f_{B}(b,w)= \widehat f_A(b-\mathbf {\widehat g}(w))
\end{align}
where $\widehat f_A$ is a
sieve minimum distance estimator of $f_A$ given by
\begin{equation*}
\widehat f_{A}\in\argmin_{f\in\mathcal{A}_K}\Big\{%
\int_{\R^{d-1}}\int_\mathbb R\big|(\mathcal{F}
f)(t, tx) - \widehat h(x,t;\widehat g)\big|^2d\nu(t)dx\Big\}
\end{equation*}
and $\mathcal{A} _K$ is a sieve space of dimension $K=K(n)<\infty $ with basis
functions $\{q_k\}_{k\geq 1}$. The sieve dimension $K=K(n)$ grows slowly with sample size $n$. 
Sieve estimation is also convenient to impose additional constraints on the unknown functions. These constraints, such as positivity, can be directly imposed on the sieve space $\mathcal A_K$. When the constraints are not binding, we may consider $\mathcal A_K$ without constraints,  which then coincides with the linear sieve space $\mathcal{A} _K=\big\{\phi(\cdot)=\beta' q^K(\cdot):\,\beta\in\R^K\big\}$ where $q^K(\cdot)=\big(q_1(\cdot),\dots,q_K(\cdot)\big)'$.

The unknown conditional characteristic function $h$ is replaced by the plug-in series least squares estimator
\begin{equation}\label{est:h}
\widehat{h}(x,t; \widehat g)= p^K(x)'\widehat P^{-1}\frac{1}{n}\sum_{j=1}^n\exp\big(it(Y_j-\widehat g(S_j))\big)p^K(X_j),
\end{equation}
where $p^K(\cdot)=\big(p_1(\cdot),\dots,p_K(\cdot)\big)'$ is a vector of basis functions and  we use  the notation $\widehat P= n^{-1}\sum_{j=1}^np^K(X_j)p^K(X_j)'$. Thus, the same sieve dimension $K$ (as for $\mathcal A_K$) is used to approximate the conditional characteristic function $h$. 
The estimator $\widehat g$ of  the regression function $g$ is based on the conditional mean restriction \eqref{mod:vcim}. We do not impose an explicit form of this estimator but rather impose a rate condition on the estimator  $\widehat g$  to obtain our asymptotic results. In particular, $\widehat g$ can account for generated regressors when $W$ are control function residuals. 
Below, Example \ref{exmp:series:vc} provides an illustration of estimating $g$ via series least squares. Although estimation of the density $f_B$ involves two preliminary steps (estimation of $h$ and $g$) it should be emphasized that the estimation procedure is equivalent to a one-step sieve minimum distance estimator, which, additionally involves the conditional characteristic equation $h(X,t;g)=\E\big[\exp\big(it(Y-g(S))\big)\big|X\big]$ and the conditional moment equation \eqref{mod:vcim}.

When no constraints are imposed, the sieve minimum distance estimator \eqref{est:denst:A} is of closed form. In this case, $\widehat f_{ B}$ coincides with the {\it double series least squares estimator}
\begin{align*}
\widehat f_{ B}(b,w)=q^K(b-\widehat{\mathbf g}(w))'Q^{-1}\int_{\R^{d-1}}\int_\R \big(\mathcal F q^K\big)(-t, -tx)\, 
\widehat
h(x,t; \widehat g)d\nu(t)\,dx
\end{align*}
where
\begin{align}\label{def:mat:Q}
Q=\int_{\R^{d-1}}&\int_\R\big(\mathcal F q^K\big)(-t, -tx)\big(\mathcal F q^K\big)(t, tx)'d\nu(t)dx
\end{align}
is assumed to be nonsingular (at least for $K$ sufficiently large). Nonsingularity of $Q$ is satisfied by Hermite functions (under mild conditions) which, as the following example illustrates, are a convenient choice of bases for the sieve space $\mathcal A_K$. 

\begin{example}[Hermite Functions]\label{exmp:hermite}
Consider a linear sieve space $\mathcal{A} _K$ spanned by Hermite
functions (which are orthonormalized Hermite polynomials) given for $k= 0,1,2,\dots$ by 
\begin{align*}
q_{k+1}(t)=\frac{(-1)^k}{\sqrt{2^k k!\sqrt\pi}}\exp(t^2/2)\frac{d^k}{dt^k}%
\exp(-t^2).
\end{align*}
These functions form an orthonormal basis in $L^2(\mathbb{R})=\set{\phi:\int_\R \phi^2(t)dt<\infty}$ and are convenient in our framework: Hermite functions are eigenfunctions of the Fourier transform satisfying 
$( \mathcal{F }q_{k+1})(t)/\sqrt{2\pi}= i^k q_{k+1}(t)=:\widetilde q_{k+1}(t)$.
Thus, the double series least squares estimator of $f_ B$ given in \eqref{est:denst:A}
simplifies to 
\begin{equation}\label{est:denst:hermite}
\widehat f_B(\cdot,w)=q^K(\cdot-\widehat{\mathbf g}(w))'Q^{-1}\int_{\R^{d-1}}\int_\R \widetilde q^K(-t, -tx)\widehat
h(x,t; \widehat g) d\nu(t)\,dx,
\end{equation}
where $Q=(2\pi)^{d/2}\int_{\R^{d-1}}\int_\R \widetilde  q^K(-t, -tx)\widetilde  q^K(t, tx)'d\nu(t)dx$. 
The implementation of this estimator is straightforward. In the finite sample analysis, we also estimate the conditional characteristic function $h$ using Hermite functions. 
\end{example}
We now consider the case of tensor product basis functions. We make use of the notation $\text{I}_{K_1}$ for the $K_1$-dimensional identity matrix and $\otimes$ for the Kronecker product.
\begin{lemma}\label{lem:vrs:ident}
Let $\mathcal A_K$ be a linear sieve space  spanned by tensor product Hermite functions $q^K(t, u)=q^{K_0}(t)\otimes q^{K_1}(u)$ where $K=K_0K_1$.
Then, the matrix $Q$ given in \eqref{def:mat:Q} simplifies to
\begin{align*}
Q=Q_0\otimes \text{I}_{K_1}\quad  \text{ where }\quad Q_0=(2\pi)^{d/2}\int_\R  |t|^{1-d} \widetilde q^{K_0}(-t) \widetilde q^{K_0}(t)'d\nu(t).
\end{align*}
\end{lemma}
Lemma \ref{lem:vrs:ident} shows that given tensor product Hermite functions in a linear sieve space only the dimension parameter $K_0$ used to approximate the random intercept induces potentially small eigenvalues of $Q$ and hence,  slow accuracy of estimators. Thus, if we are willing to restrict the complexity of sieve estimation in terms of $K_0$, for instance, by imposing a fixed dimension $K_0$, then the rate of convergence does not suffer from an ill-posed inverse problem. This semiparametric specification provides a numerically stable estimation procedure as we see in our Monte Carlo simulation section. 
%

\begin{remark}[VRS Density Estimation]\label{remark:ind}
The proposed methodology implies closed form estimators of densities of subvectors of the VRC vector.
Consider the setup of Lemma \ref{lem:vrs:ident} with a linear sieve space and $q^K(t, u)=q^{K_0}(t)\otimes q^{K_1}(u)$ are Hermite functions.
The double series least squares estimator \eqref{est:denst:hermite} yields the estimator of the VRS density $ f_{B_1}(\cdot,w)$ given by
\begin{align}\label{VRS:estimator}
\widehat f_{B_1}(\cdot,w)
=q^{K_1}(\cdot-\widehat{\mathbf g}_1(w))' \int_{\R^{d-1}}\int_\R b_{K_0}(t)\widetilde q^{K_1}(-tx)\widehat
h(x,t; \widehat g) d\nu(t)\,dx
\end{align}
where
\begin{align*}
b_{K_0}(t):=\int_\R q^{K_0}(a)'da\,Q_0^{-1}\widetilde q^{K_0}(-t).
\end{align*}
Here, $K_1$ coincides with the dimension of basis functions used for the estimator $\widehat h$.
Note that we impose restrictions on the weighting measure $\nu$ below such that $\int_\R |t|^{1-d} d \nu(t)$ is bounded above and away from zero and hence, $Q_0$ is well defined.
\end{remark}

\begin{example}[Series estimation of the varying coefficients]\label{exmp:series:vc}
We provide an explicit estimator of the mean regression function $g(s)=\E[Y|S=s]$. Series estimation of $g$ is convenient, in particular, as an additive structure of $g$ can be easily imposed.
To do so, consider  the basis functions $p_k$, $k\geq 1$, and introduce the vector $p_d^K(s)=\big(p^K(w)',x_1p^K(w)',\dots,x_{d-1}p^K(w)'\big)'\in\mathbb R^{dK}$.
The series least squares estimator of $g$ is given by
\begin{align*}
\widehat g(s)= p_d^K(s)'\widehat P_d^{-1}\sum_{j=1}^nY_jp_d^K(S_j),
\end{align*}
where $\widehat P_d= n^{-1}\sum_{j=1}^n p^K_d(S_j)p^K_d(S_j)'$. 
 For instance, when using B-Splines as basis functions, for asymptotic properties of such estimators see \cite{xia1999}, \cite{fan2003},  \cite{xue2012}, or \cite{ma2015}. 
\end{example}

\paragraph{Notation.} 
Introduce the norms $\|\phi \| _{\nu }=\big(\int_{\R} \int_{\R^{d-1}}|\phi
(x,t)|^{2}dx d\nu(t)\big)^{1/2}$ and $\| \phi\|_{\R^d}=\big(\int_{\R^d}|\phi
(u)|^{2}du\big)^{1/2}$ with associated Hilbert spaces $L^2_\nu=\set{\phi:\|\phi\|_\nu<\infty}$ and $L^2(\R^d)=\set{\phi:\|\phi\|_{\R^d}<\infty}$.
 Moreover, $\|\cdot\|$ and $\|\cdot\|_\infty$ denote the $\ell_2$--norm and the supremum norm.  For a matrix $A$, 
the operator norm is denoted by $\|A\|$. For a random vector $V$ the corresponding calligraphic capital letter $\mathcal V$ denotes its support. Let $\Pi_K:L^2(\R^d)\to\mathcal A_K$ denote the  least squares projection onto $\mathcal A_K$, i.e., $\Pi_K f=\argmin_{\phi\in\mathcal A_K}\|\phi-f\|_{\R^d}$ for all $f\in L^2(\R^d)$. Further, define the function $\gamma:\R\to\R^K$ given by
$\gamma(t)=P^{-1}\E[\exp\big(it(Y-g(S))\big)p^K(X)]$ where we recall  $P=\E[ p^K(X)p^K(X)']$.
We use the notation $a_n\sim b_n$ for $c b_n\leq a_n\leq C b_n$ given two
constant $c,C>0$ and all $n\geq 1$.
\subsection{Rate of Convergence}\label{sub:sec:rate}
This subsection provides rates of convergence of the proposed estimators. In particular, we see that only the rate of the VRC density hinges on the minimal eigenvalues of $Q$. On the other hand, it is shown below that the rate of convergence of the estimator of the varying random slope $B_1$ is not affected by the choice of $\nu$.

We introduce an assumption required for our asymptotic results. Below, $\lambda_{\max}(M)$ denotes the maximal eigenvalue of a matrix $M$. Throughout the paper, $(\tau _k)_{k\geq 1}$ and $(\lambda _{k})_{k\geq 1}$ denote nonincreasing sequences. 
\begin{ass}
\label{Ass:bas} (i) A random sample $\set{(Y_{i},X_{i},W_i)}_{i=1}^n $ 
of $(Y,X,W)$ is observed. 
(ii) The   measure $\nu$ satisfies $0<\int_\R |t|^{1-d} d\nu(t)<\infty$,  $\int_\R t^2 d\nu(t)<\infty$ and $\|h(\cdot,\cdot;g)\|_\nu<\infty$.
(iii) $\lambda_{\max}\big(\int_{\R^{d-1}}p^K(x)p^K(x)'dx\big)=O(1)$ and $\lambda_{\max}\big(\int_{\R^{d}}q^K(a)q^K(a)'da\big)=O(1)$.
(iv) $\lambda_{\max}(\lambda_KP^{-1})=O(1)$ and $\sup_{x\in \R^{d-1}}\| p^K(x)\|^2=O(K)$ satisfying 
$K\log(n)=o(n\lambda _K)$.
(v)   $\lambda_{\max}(\tau _K Q^{-1})=O(1)$ and  $\|\mathcal{F}
(\Pi_K f_A-f_A)\|_\nu^2=O( \tau _K\|\Pi_K f_A-f_A\|_{\mathbb R^d}^2)$.
\end{ass}
Assumption \ref{Ass:bas} (ii) imposes a mild restriction on the weighting measure $\nu$ and is satisfied, for instance, if $\nu$ is the log-normal distribution.
The maximal eigenvalue restriction imposed in Assumption \ref{Ass:bas} (iii)  is automatically satisfied if  $\{p_k\}_{k\geq 1}$ and $\{q_k\}_{k\geq 1}$ form orthonormal bases in $L^2(\R^{d-1})$ and $L^2(\R^d)$, respectively.
Assumption \ref{Ass:bas} (iv) allows the minimal eigenvalues of $P$ to tend to zero, see also  \cite{chen2015}, who consider a similar restriction on the growth of basis functions. This condition holds for Hermite functions, which are uniformly bounded, but also for polynomial
splines, Fourier series and wavelet bases, see also \cite{belloni2012} for further discussion and extensions of \cite{newey1997}. 
 It is well known, that the smallest eigenvalue of $P$ is uniformly bounded away from zero if $\set{p_l}_{l\geq 1}$ forms an orthonormal basis and $f_X$ is uniformly bounded away from zero on the support of $X$. Thus, Assumption \ref{Ass:bas} (iv) is particularly suited in our case where a large support assumption of $X$ is required for identification. Here, $\lambda_K$ has the interpretation of a truncation parameter in estimation problems to ensure that the denominator is bounded away from zero,  see \cite{breunig2016}.
Assumption \ref{Ass:bas} (v) imposes a rate restriction on the minimal eigenvalue of the
matrix  $Q$ relative to $(\tau _k)_{k\geq 1}$. Assumption \ref{Ass:bas} (v) further specifies a link of the sieve approximation error on the RC density $f_A$ between the ``strong'' norm $\| \cdot\|_{\R^d}$ and the ``weak'' norm $\|\cdot\|_\nu$. Note that Assumption \ref{Ass:bas} (v) is only required for estimating the VRC density but not for estimating the VRS density. Such stability conditions are commonly imposed in the literature on nonparametric instrumental variable estimation, see, for instance, \cite{BCK07econometrica} or \cite{Chen08}, but rely on mapping properties of an unknown conditional expectation operator.  In addition, in VRC models it is possible to provide primitive conditions, as we see below.

\begin{proposition}
\label{prop:eig} Let $\widetilde\nu(t):=|t|^{1-d} (d\nu/d\mu)(t)$ where $\mu$ denotes the Lebesgue measure and assume that $\{q_k\}_{k\geq 1}$ is an orthonormal basis in $L^2(\mathbb R^d)$.
\begin{itemize}
\item[(i)]Suppose that, for some constant $0<c<1$,
for all $n\geq 1$ and any non-zero vector $a\in\mathbb{R}^K$ the inequality 
\begin{equation}  \label{prop:eig:cond}
\int_\mathbb R \big| a^{\prime }(\mathcal F q^K)(t)\big|^2\1\nolimits{%
\left\lbrace \widetilde\nu(t)<\tau_K\right\rbrace }\mu(dt)\leq c\int_\mathbb R \big|a^{\prime }(\mathcal F q^K)(t)\big|^2\mu(dt)
\end{equation}
holds. Then, the smallest
eigenvalue of $\tau_K^{-1}Q$ is bounded away from zero and bounded above.

\item[(ii)] 
If 
\begin{equation}\label{prop:eig:cond:bias}
\sum_{l>K}\int_\mathbb R\big|(\mathcal F q_l)(t) \big|^2\,\widetilde \nu(t)d\mu(t)=O(\tau_K)
\end{equation}
 then $\|\mathcal{F}
(\Pi_K f_A-f_A)\|_\nu^2=O\big(\tau _K\|\Pi_K f_A-f_A\|_{\mathbb R^d}^2\big)$.
\end{itemize}
\end{proposition}
In the case of Hermite functions, condition \eqref{prop:eig:cond} requires  $\nu$ to be sufficiently heavy tailed while condition \eqref{prop:eig:cond:bias} requires  $\int_\mathbb R\big|(\mathcal F q_l)(t) \big|^2\,\widetilde \nu(t)d\mu(t)$ to be sufficiently small for all $l>K$. Both conditions impose restrictions on the weighting measure $\nu$ and the basis functions $\{q_k\}_{k\geq 1}$ but do not  on the unknown density $f_A$ itself.
The following example discusses the  log-normal distribution as a weighting measure and provides primitive conditions for inequality \eqref{prop:eig:cond}.
\begin{example}[Log-normal weighting]\label{Exmp:LN}
The weighting measure $\nu$ is given by the distribution
$\textsl{Lognormal}(0,\sigma^2)$ for some $\sigma>0$. 
Thus, the function $\widetilde \nu$ as introduced in Proposition \ref{prop:eig} coincides with 
$\widetilde \nu(t)=\exp\big(-(\log t)^2/(2\sigma^2)\big)/(\sqrt{2\pi}\sigma \, t^d)$ for all $t\geq 0$. 
In this case, $\widetilde \nu(t)<\tau_K$ if and only if
$\exp\big(-(\log t)^2/(2\sigma^2)\big)<\tau_K \sqrt{2\pi}\sigma t^d$
which holds for all $t\geq 0$ satisfying
\begin{align*}
(\log t)^2/(2\sigma^2)+d(\log t)+\log\big(\sqrt{\pi}\sigma\tau_K\big)>0.
\end{align*}
Note that Hermite functions have most of its support close to zero when $K$ is chosen moderately and 
hence, inequality \eqref{prop:eig:cond} is always satisfied for $\tau_K$ sufficiently small. The decay of the eigenvalues of $Q$ can be directly computed via numerical approximation of the integral.
\end{example}
We impose  regularity conditions on the unknown functions of interest.
To do so, we introduce the space of square integrable functions $L_S^2=\big\{\phi:\|\phi\|_S:=\sqrt{\E\phi^2(S)}<\infty\big\}$. Further, we define the class of functions $\mathcal G:=\mathcal G(n):=\{\phi=\beta'p_d^K:\|\phi-g\|_\infty\leq \sqrt{K/n}\}$.
The covering number $N(\mathcal G , \|\cdot\|_S, \epsilon)$ is the minimal number of $L_S^2$--balls of radius $\epsilon$ needed to cover $\mathcal G$.
\begin{ass}
\label{Ass:lin} (i) $\| \Pi_K f_A-f_A\|_{\mathbb R^d}=O(K^{-\zeta/d})$ and $\|\gamma'p^K-
h\|_\nu=O(K^{-\rho/(d-1)})$.
 (ii) $f_A$ is continuously differentiable with square integrable Jacobian matrix $Df_A$.  (iii) 
$\|\widehat g-g\|_\infty^2=O_p(K/n)$ and  $\|\widehat {\mathbf g}(w)-\mathbf g(w)\|^2=O_p\big(K/(n\tau_K\lambda_K)\big)$ for all $w\in\mathcal W$. 
(iv) There exists $C_n>0$ with $C_n \sqrt{K/n}=O(1)$ such that $\int_0^1\sqrt{1+\log N(\mathcal G, \|\cdot\|_S, \epsilon)}d\epsilon\leq C_n<\infty$.
\end{ass}
Assumption \ref{Ass:lin} (i) captures regularity conditions on the density $f_A$ and the conditional characteristic function $h$ via sieve approximation
errors. Example \ref{ex:hermite:smooth} characterizes the sieve approximation condition when using  Hermite basis functions relative to the smoothness of the unknown function of interest. For further discussion and other examples of sieve bases, see \cite{Chen07}. 
 Assumption \ref{Ass:lin} (ii) imposes  a mild smoothness condition on the density $f_A$. 
Assumption \ref{Ass:lin} (iii) imposes rate restrictions on the varying coefficients estimator. 
Assumption \ref{Ass:lin} (iv) is a regularity condition on the function class $\mathcal G$ and was also imposed in the literature, for instance, in \cite[Lemma C.3 (i)]{Chen08}.
\begin{example}[Hermite functions  (cont'd)]\label{ex:hermite:smooth}
Consider again the case where the linear sieve space $\mathcal A_K$ is spanned by Hermite functions.
 If $f_A$ has $2\zeta$ derivatives such that $\int f^{(\iota)}(x)dx<\infty$ for all $\iota\leq 2\zeta$ then the sieve approximation condition $\| \Pi_K f_A-f_A\|_{\mathbb R^d}=O(K^{-\zeta/d})$ in  
 Assumption \ref{Ass:lin} (i) 
 is automatically satisfied 
due to Lemma 2 of \cite{coppejans2002}. See also \cite{bongioanni2009} where a Sobolev space for Hermite functions is constructed  (and also in the case of Laguerre polynomials). On the other hand, if $A$ has compact support and $f_A$ belongs to a Sobolev ellipsoid of smoothness $\zeta$ then the sieve approximation error $O(K^{-\zeta/d})$ is obtained for B-splines, wavelets, or trigonometric basis functions, see \cite{Chen07}.
\end{example}

Consider first estimation of the joint density of the VRC vector $B$ holding observed characteristics fixed. The following theorem provides the rate of convergence in  $L^2(\R^d)$--norm of the estimator $\widehat f_{B}(\cdot,w)$ given in \eqref{est:denst:A} for some fixed $w\in\mathcal W$.
\begin{theorem}\label{thm:rate:fdb}
Let Assumptions \ref{ass:ind}--\ref{Ass:lin} be satisfied. Then for any $w\in\mathcal W$:
\begin{align*}
\int_{\R^d}\big|\widehat f_B(b,w)-f_B(b,w)\big|^2db=O_p\left((\tau_K\lambda_K)^{-1}\Big(
n^{-1}K+K^{-2\rho/(d-1)}\Big)+K^{-2\zeta/d}\right).
\end{align*}
\end{theorem}
We see from the previous result that the parameters $\tau_K$ and $\lambda_K$, which capture the minimal eigenvalues of $Q$ and $P$, respectively, slow down the rate of convergence. As the weighting measure $\nu$ influences the eigenvalues of $Q$ it does also affect the rate of convergence of the joint density of the varying random coefficient  $B$. Note that the optimal choice of the dimension parameter $K$ is lower than for usual nonparametric estimation problems and relative to the value of $\tau_K$ and $\lambda_K$. 

We now provide the rate of convergence in  $L^2(\R^d)$--norm of $\widehat f_B$ when the dimension parameter $K$ is chosen such that the variance part $(n\tau_K\lambda_K)^{-1}K$ and the squared bias part $K^{-2\zeta/d}$ are of the same order. For ease of exposition we assume that $\lambda_K$ is uniformly bounded away from zero, which requires a sufficiently large support of regressors $X$. We further consider the case that $n^{-1}K$ dominates the squared bias part of estimating $h$ which is $K^{-2\rho/(d-1)}$. 
The next result immediately follows from Theorem \ref{thm:rate:fdb} and hence, its proof is omitted.
\begin{corollary}\label{cor:rate:fdb}
Let Assumptions \ref{ass:ind}--\ref{Ass:lin} be satisfied. Further if $\lambda_K$ is uniformly bounded away from zero, $\tau_K\sim K^{-2\alpha/d}$, $K\sim n^{d/{(2\zeta+2\alpha+d)}}$ and if
\begin{align}\label{cond:rate:fdb}
\frac{\alpha+\zeta}{d}\leq \frac{\rho}{(d-1)}
\end{align}
holds, then for any $w\in\mathcal W$:
\begin{align*}
\int_{\R^d}\big|\widehat f_B(b,w)-f_B(b,w)\big|^2db=O_p\Big(n^{-2\zeta/(2\zeta+2\alpha+d)}\Big).
\end{align*}
\end{corollary}

As we see from the previous result, the usual rate of convergence for ill-posed inverse problems is obtained if $K$ levels variance and squared bias, see also \cite{hoderlein2010} in the case of kernel estimation in ordinary RC models. Yet in contrast to \cite{hoderlein2010} (see their Section 4.3), no heavy-tailedness of covariates $X$ is required for our convergence results. In particular, \cite{hohmann2016} showed under Gaussian white noise that the rate of convergence can be much slower, even severely ill-posed, under lighter tails of $X$. 
Also note that the condition \eqref{cond:rate:fdb} ensures that $n^{-1}K$ dominates the squared bias of estimating $h$ and is only a mild restriction. Indeed, if, for instance, $d=2$ then the usual rate of convergence is obtained if $2\rho\geq \zeta+\alpha$. 

Below,  we establish a rate of convergence for estimating the density of the varying random slope $B_1$ holding observed characteristics $W$ fixed. The next results show that for estimating the VRS vector $B_1$ when we restrict the sieve dimension, to approximate the random intercept, to be bounded. In this case the rate of convergence does not depend on the eigenvalue behavior of the matrix $Q$. 
\begin{theorem}\label{thm:rate:fdb1}
Let Assumptions \ref{ass:ind}, \ref{large:supp}, \ref {Ass:bas} (i)--(iv), and \ref{Ass:lin} be satisfied where $q^K(t, u)=q^{K_0}(t)\otimes q^{K_1}(u)$ are tensor product Hermite functions with $K_0=O(1)$.
Then the closed form estimator $\widehat f_{B_1}(\cdot,w)$ given in \eqref{VRS:estimator} satisfies for any $w\in\mathcal W$:
\begin{equation*}
\int_{\R^{d-1}}\big|\widehat f_{B_1}(b_1,w)-f_{B_1}(b_1,w	)\big|^2db_1=O_p\Big(\lambda_{K_1}^{-1}\Big(n^{-1}K_1+K_1^{-2\rho/(d-1)}\Big)+K_1^{-2\zeta/(d-1)}\Big).
\end{equation*}
\end{theorem}
Note that the condition $K_0=O(1)$ together with Assumption \ref{Ass:lin} (i), that is, $\| \Pi_K f_A-f_A\|_{\mathbb R^d}=O(K^{-\zeta/d})$, imposes a semiparametric restriction on the  RC density $f_A$. We emphasize that the parametric restriction concerns only the random intercept $A_0$ but not the random slope $A_1$. 
Also the rate of convergence derived in the previous result depends on the dimension parameter $K_1$ only.
This is due to the definition of our estimator where the dimension of basis functions used to estimate the conditional characteristic function $h$ coincides with the dimension parameter $K_1$, see also Remark \ref{remark:ind}. Similarly to Theorem \ref{thm:rate:fdb1}, convergence rates of subvector VRS densities can be derived.

The following result provides the rate of convergence when the dimension parameter $K_1$ such that the variance part $(n\lambda_{K_1})^{-1}K$ and the squared bias $K_1^{-2\zeta/d}$ are of the same order and $n^{-1}K_1$ dominates the squared bias of estimating $h$. Again we consider the case where $\lambda_{K_1}$ is uniformly bounded away from zero. The next result immediately follows from Theorem \ref{thm:rate:fdb1} and hence, its proof is omitted.
\begin{corollary}\label{cor:rate:fdb1}
Let the conditions of Theorem \ref{thm:rate:fdb1} be satisfied. Further, if $\lambda_{K_1}$ is uniformly bounded away from zero, $K_1\sim n^{(d-1)/{(2\zeta+d-1)}}$ and $\zeta\leq \rho$
holds, then the closed form estimator $\widehat f_{B_1}(\cdot,w)$ given in \eqref{VRS:estimator} satisfies for any $w\in\mathcal W$:
\begin{align*}
\int_{\R^{d-1}}\big|\widehat f_{B_1}(b,w)-f_{B_1}(b,w)\big|^2db=O_p\Big(n^{-2\zeta/(2\zeta+d-1)}\Big).
\end{align*}
\end{corollary}
As we see from the previous result, the rate of convergence coincides with the usual optimal rate in nonparametric density of dimension $d-1$. We see that the rate for estimating the density of the varying random slope $B_1$ is not affected by the ill-posedness under the condition $K_0=O(1)$. 

\subsection{Pointwise Limit Theory for Linear Functionals}\label{sub:sec:pw:inference}
This subsection is about inference  of linear functionals $\ell(\cdot):L^2(\mathbb R^d)\to\mathbb R$ of the density $f_B(\cdot,w)$ for some realization $w$ of $W$. Examples of linear functionals are, but are not limited to, the point-evaluation functional or (weighted) averages of $f_B(\cdot,w)$. The linear functional $\ell(f_B(\cdot,w))$ is estimated below using the plug-in sieve estimate estimator $\ell(\widehat f_B(\cdot,w))$. This subsection provides a limit theory for this plug-in estimator. 

As mentioned in the introduction, an important example of a linear functional of $f_B$ is 
the density of the potential outcome $Y^{s}=\mathbf g(w)'x+A'x$, where $s=(x',w')'$. Indeed, the density of  $Y^s$ can be written as\footnote{This relation follows immediately by the well known property $f_{A_0+A_1'x}(c)=\int_{\mathbb R^{d-1}} f_A\big(c-u'x,u\big)du$.}
\begin{align*}
f_{ Y}(y,s)=\int_{\mathbb R^{d-1}} f_B\big((y-x'b_1,b_1),w\big)\, db_1.
\end{align*}
Also the  weighted averages of potential outcomes $\int_\mathbb R  \omega(y) f_Y\big(y,s\big)dy$
for some function $\omega$ is a linear functional of $f_B(\cdot,w)$. In particular, we obtain the distribution of the potential outcomes which is relevant also in the context of quantile treatment effects.

The asymptotic distribution result below requires the following additional notation and assumptions. We introduce the sieve variance
\begin{align}\label{sieve:var}
\textsl{v}_K(w)&=\ell\big(q^K(\cdot -\mathbf g(w))\big)' Q^{-1/2}\,\Sigma\,Q^{-1/2}\ell\big(q^K(\cdot -\mathbf g(w))\big)
\end{align}
where 
\begin{align*}
\Sigma=\int_\R\int_\R  R(s) P^{-1}\E\Big[p^K(X)\rho(s)\rho(-t)p^K(X)'\Big] P^{-1}R(-t)' d\nu(s)d\nu(t),
\end{align*}
using the notation $\rho(t)=\exp(it(Y-g(S)))-h(X,t;g)$ and the matrix valued function $R(t)=Q^{-1/2}\int_{\R^{d-1}} (\mathcal F q^K)(t, tx)p^K(x)'dx$. 
We replace the sieve variance $\textsl{v}_K$ by the estimator
\begin{align}\label{sieve:var:est}
\widehat{\textsl{v}}_K(w)&=\ell\big(q^K(\cdot -\widehat{\mathbf g}(w))\big)' Q^{-1/2}\, \widehat \Sigma\, Q^{-1/2}\ell\big(q^K(\cdot -\widehat{\mathbf g}(w))\big),
\end{align}
where the matrix $\Sigma$ is replaced by
\begin{align*}
\widehat \Sigma=\int_\R\int_\R  R(s) \widehat P^{-1} \frac{1}{n}\sum_{j=1}^n p^K(X_j)\widehat \rho_j(s)\widehat \rho_j(-t) p^K(X_j)'\widehat P^{-1}R(-t)'d\nu(s)d\nu(t),
\end{align*}
with  $\widehat \rho_j(t)=\exp(it(Y_j-\widehat g(S_j)))-\widehat h(X_j,t;\widehat g)$.
In order to derive the asymptotic distribution of our estimator we require addition assumptions. Below, we denote the inverse Fourier transform by $(\mathcal F^{-1}\phi)(u)=(2\pi)^{-d}\int_{\R^d}\exp(-itu)\phi(t)dt$.
\begin{ass}\label{A:inf}
 (i)  The minimal eigenvalue of $\lambda_K^{-1}\Sigma$ is uniformly bounded away from zero. 
 (ii) It holds $C_n^2K^2=o(n\lambda_K)$. 
 (iii) For all $w\in\mathcal W$: $\sqrt{n}\ell\big(\Pi_K f_A(\cdot -\mathbf g(w))-f_A(\cdot -\mathbf g(w))\big)=o\big(\sqrt{\textsl{v}_K(w)}\big)$, 
 $\sqrt n \ell\big([\mathcal F^{-1}(\gamma'p^K- h)](\cdot -\mathbf g(w))\big) =o\big(\sqrt{\textsl{v}_K(w)}\big)$, and
 $\sqrt n \big\|\widehat{\mathbf g}(w)-\mathbf g(w)\big\|=o_p\big(\sqrt{\textsl{v}_K(w)}\big)$.
 (iv) The sieve space $\mathcal A_K$ is linear:  $\mathcal{A} _K=\big\{\phi(\cdot)=\beta' q^K(\cdot):\,\beta\in\R^K\big\}$ and $q_l$, $l\geq 1$, are continuously differentiable. 
(v) It holds $\int_\R \| R(t)\|^2d\nu(t)=O(1)$. 
\end{ass}
Assumption \ref{A:inf} (i) implies a lower bound on the sieve variance which we require to achieve asymptotic distribution results of the estimator. 
This condition implies that the sieve variance attains the lower bound $\textsl{v}_K(w)\geq C_\Sigma\, \|\ell\big(q^K(\cdot -\mathbf g(w))\big)' Q^{-1/2}\|^2/\lambda_K$ for some constant $C_\Sigma>0$. 
For instance, when $\ell(\cdot)$ is the point evaluation functional  and $Q$ is a diagonal matrix with polynomial decay of order  $-2\alpha/d$ we obtain $\textsl{v}_K(w)\geq C_\Sigma\,  K^{(2\alpha+d)/d}$ provided that $\lambda_K$ is bounded from below and $\|q^K(a -\mathbf g(w))\|^2\geq K$ which holds at most points $a\in\mathcal A$ (see \cite{belloni2012}). Consequently the lower bound of the sieve variance corresponds to the mildly ill-posed case in \cite[p. 1053]{chen2013}.  
Assumption \ref{A:inf} (ii) imposes a stronger rate requirement on $K$ than the one required for consistency. 
 Assumption \ref{A:inf} (iii) specifies a pointwise sieve approximation error and has the interpretation of an undersmoothing condition, which is required to ensure that the approximation bias is asymptotically negligible. 
 Assumption \ref{A:inf} (iv) restricts the sieve space to be linear and, in particular, that constraints on density estimation, such as positivity, are not binding. If such constraints are binding asymptotically then such shape restriction can lead to non-normal distributions.
 Assumption \ref{A:inf} (v) is automatically satisfied if the basis under consideration is  given by Hermite functions.
 
 \begin{ass}\label{A:inf:var}
  It holds $\zeta_n=o(1)$ where
  $\zeta_n=\sqrt{K
\log(n)/(n\tau_K\lambda_K)}+  K^{1/2-\rho/(d-1)}$ and $\|\gamma'p^K-
h\|_\nu=O(K^{-\rho/(d-1)})$.
 \end{ass}
 Assumption \ref{A:inf:var} strengthens the rate conditions imposed in Assumption \ref{A:inf} and is required for consistent estimation of the sieve variance $\textsl{v}_K(w)$.
 The next result establishes the asymptotic distribution of the estimator $\ell(\widehat f_B(\cdot,w)) $.
\begin{theorem}\label{thm:inference}
  Let Assumptions \ref{ass:ind}--\ref{Ass:bas}, \ref{Ass:lin} (ii)--(iv), and  \ref{A:inf} be satisfied. Then,  for any $w\in\mathcal W$:
    \begin{align*}
      \sqrt{n/\textsl{v}_K(w)}\,\Big(\ell\big(\widehat f_B(\cdot,w)\big) - \ell\big(f_B(\cdot,w)\big)\Big)\stackrel{d}{\rightarrow}\mathcal N(0,1).
    \end{align*}
If, in addition Assumption \ref{A:inf:var} holds, then
      \begin{align*}
      \sqrt{n/\widehat{\textsl{v}}_K(w)}\,\Big(\ell\big(\widehat f_B(\cdot,w)\big) - \ell\big(f_B(\cdot,w)\big)\Big)\stackrel{d}{\rightarrow}\mathcal N(0,1).
    \end{align*}
\end{theorem}
A direct implication of the previous results concerns inference on functionals of the VRS density $f_{B_1}(\cdot,w)$.
The next result is based on plug-in series estimator  \eqref{VRS:estimator} using tensor product Hermite functions.
  In this case, the sieve variance simplifies to
 \begin{align*}
\textsl{v}_{K_1}(w)&=
\ell\big(q^{K_1}(\cdot -\mathbf g(w))\big)' \Sigma_1\ell\big(q^{K_1}(\cdot -\mathbf g(w))\big)
\end{align*}
 where the matrix $\Sigma_1$ is coincides with $\Sigma$ except that $K$ replaced by $K_1$ and $R(t)$ replaced by $\int_{\R^{d-1}} b_{K_0}(t)\widetilde q^{K_1}(tx)p^{K_1}(x)'dx$, where the function $b_{K_0}$ is given in Remark \ref{remark:ind}.
Now if $K_0=O(1)$ a lower bound for the sieve variance is $\textsl{v}_{K_1}(w)\geq C_{\Sigma_1} \,\|\ell\big(q^{K_1}(\cdot -\mathbf g(w))\big)\|^2/\lambda_{K_1}$ for some constant $C_{\Sigma_1}>0$. When $ \lambda_{K_1}$ is uniformly bounded away from zero this corresponds to the usual lower bound in well posed estimation problems, see \cite{newey1997} or \cite{belloni2012}. 
An estimator $\widehat{\textsl{v}}_{K_1}(w)$ of the sieve variance $\textsl{v}_{K_1}(w)$  is obtained by replacing the covariance matrix $\Sigma_1$ by $\widehat \Sigma_1$ which is analog to the definition of $\widehat \Sigma$. The next result is an immediate consequence of Theorem \ref{thm:inference} and hence its proof is omitted.
\begin{corollary}\label{coro:inference}
  Let Assumptions \ref{ass:ind}, \ref{large:supp}, \ref {Ass:bas} (i)--(iv), \ref{Ass:lin} (ii)--(iv), \ref{A:inf}, and \ref{A:inf:var} be satisfied where $q^K(t, u)=q^{K_0}(t)\otimes q^{K_1}(u)$ are tensor product Hermite functions with $K_0=O(1)$. Then,  for any $w\in\mathcal W$:
    \begin{align*}
      \sqrt{n/\widehat{\textsl{v}}_{K_1}(w)}\,\Big(\ell\big(\widehat f_{B_1}(\cdot,w)\big) - \ell\big(f_{B_1}(\cdot,w)\big)\Big)\stackrel{d}{\rightarrow}\mathcal N(0,1).
    \end{align*}
\end{corollary}

\subsection{Bootstrap Uniform Confidence Bands}\label{sec:bootstrap:conf}
This subsection provides a  bootstrap procedure to construct uniform confidence bands for $f_B(\cdot,w)$ and establishes asymptotic validity of it. The multiplier bootstrap procedure is as follows. Let $(\varepsilon_1,\ldots,\varepsilon_n)$  be a bootstrap sequence of i.i.d. random variables drawn independently of the data $\{(Y_1,X_1,W_1), \ldots, (Y_n,X_n,W_n)\}$, with $\E[\varepsilon_j] = 0$, $\E[\varepsilon_j^2] = 1$, $\E[|\varepsilon_j|^{3}] < \infty$ for all $1\leq j\leq n$.
Common choices of distributions for $\varepsilon_j$ include the standard Normal, Rademacher, and the two-point distribution of \cite{mammen1993}. Further, $\mathbb{P}^*$ denotes the probability distribution of the bootstrap innovations $(\varepsilon_1,\ldots,\varepsilon_n)$ conditional on the data. For any $w\in\mathcal W$,  we introduce the bootstrap process
\begin{align*}
\mathbb Z^*(b,w) = \frac{q^K(b-\widehat{\mathbf g}(w))' Q^{-1/2}}{\sqrt{\widehat{\textsl{v}}_K(b,w)}}\left(\frac{1}{\sqrt{n}}\sum_{j=1}^n\int_\R  R(t)\widehat \rho_j(t)d\nu(t)
 \widehat P^{-1} p^K(X_j)\, \varepsilon_j
\right).
\end{align*}
Here, we use the notation $\textsl{v}_K(b,w)$ and $\widehat{\textsl{v}}_K(b,w)$ for the sieve variance and its estimator given in \eqref{sieve:var} and \eqref{sieve:var:est} in case of point-evaluation functionals.
Let $\mathcal C$ be  a closed subset of $\mathbb R^d$. 
Let $\Delta$ be the standard deviation semimetric on $\mathcal{C}$ of the Gaussian Process $\mathbb Z(b,w) = q^K(b-\mathbf g(w))'Q^{-1/2}\mathcal Z/\sqrt{\textsl{v}_K(b,w)}$ with $\mathcal Z\sim \mathcal N(0,\Sigma)$ defined as $\Delta_w(b_1,b_2) = (\E[(\mathbb Z(b_1,w) - \mathbb Z(b_2,w))^2])^{1/2}$, see \textit{e.g.} \cite[Appendix A.2]{Vaart2000}.
 To do so, we introduce the notation $N(\mathcal{C},\Delta,\epsilon)$ for the $\epsilon$-entropy of $\mathcal{C}$ with respect to norm $\Delta$. 
\begin{ass}\label{Ass:uniform}
  (i) $\mathcal{C}$ is compact and $(\mathcal{C},\Delta)$ is separable for each $n\geq 1$. (ii) There exists a sequence of finite positive integers $c_n$ such that
$ 1 + \int_0^{\infty}\sqrt{\log N(\mathcal{C},\Delta_w,\epsilon)}d\epsilon = O (c_n)$.
(iii) There exists a sequence of positive integers $r_n$ with $r_n = o(1)$ such that $K^{5/2}\lambda_K^{-2}= o(r_n^3\sqrt{n})$, $r_n c_n=O(1)$,  and
\begin{equation*}
    \lambda_K^{-1}K\sqrt{\frac{\log(n)}{n}}+\sqrt{\zeta_n}\Big(c_n+ K^{(-\rho)/(d-1)}+\sup_{b\in\mathcal C}\sqrt{\frac{n}{\textsl{v}_K(b,w)}}|\Pi_K f_B(b,w)-f_B(b,w)|\Big)=o(r_n).
   \end{equation*}
   (iv) It holds $\sup_{b\in\mathcal{C}}\left(\|q^K\big(b-\mathbf g(w)\big)'Q^{-1/2}\|^2/ \textsl{v}_K(b,w)\right) =O(1)$. 
\end{ass}
Assumption \ref{Ass:uniform} is similar to 
\cite[Assumption 6]{ChenChristensen2015} who establish asymptotic validity of uniform confidence bands in nonparametric instrumental variable estimation. 
Assumption \ref{Ass:uniform} (ii) is a mild regularity assumption, see also \cite[Remark 4.2]{ChenChristensen2015} for sufficient conditions. Assumption \ref{Ass:uniform} (iii) strengthens the rate conditions imposed on the dimension parameter $K$ and imposes a uniform sieve approximation error. 

The next theorem establishes the validity of the bootstrap for constructing uniform confidence bands for the VRC density $f_B(\cdot,w)$.
The proof of this result is based on strong approximation of a series process by a Gaussian process, and uses an anti-concentration inequality for the supremum of the approximating Gaussian process obtained in \cite{chernozhukov2014}. For sieve minimum distance estimation in nonparametric instrumental variable estimation this is also exploited by \cite{ChenChristensen2015}.
\begin{theorem}\label{thm:bands}
  Let the assumptions of Theorem \ref{thm:inference} and Assumption \ref{Ass:uniform} hold.
  Then for all $w\in\mathcal W$: 
  \begin{equation*}
      \sup_{s\in\mathbb{R}}\left|\Pr\left(\sup_{b\in\mathcal{C}}\left|\sqrt{\frac{n}{\widehat{\textsl{v}}_K(b,w)}}\Big(\widehat f_B(b,w)-f_B(b,w)\Big)\right|\leq s\right) - \Pr^*\left(\sup_{b\in\mathcal{C}}\left|\mathbb Z^*(b,w)\right| \leq s \right)\right| = o_p(1).
  \end{equation*}
\end{theorem}
 Theorem \ref{thm:bands} establishes consistency of the  sieve score bootstrap procedure for estimating the critical values of the uniform sieve t-statistic process for the VRC model. The result also contributes to the literature on ordinary RC models, as up to now, only asymptotic validity of pointwise confidence intervals is established
 and complements the results of \cite{dunker2017} who proposed tests for qualitative features of the ordinary RC's density.
\section{Simulation Studies and an Empirical Illustration}\label{sec_MC}
This section provides a finite sample analysis of the proposed estimator. 
Subsection \ref{s_Monte_Carlo_simulation} presents the finite sample performance of the proposed estimator in Monte Carlo simulations. Subsection \ref{s_Application} applies the procedure to analyze heterogeneity in income elasticity of the willingness to pay for rent in an empirical illustration. 
\subsection{Simulation Studies}\label{s_Monte_Carlo_simulation} 
This subsection presents the  finite-sample
performance of the estimator of the varying random slope in a Monte Carlo
simulation study. The experiments use a sample size of $1000$ and $1000$
Monte Carlo replications in each setting.
In each experiment, i.i.d. draws of regressors $(X,W)$ are generated from
\begin{equation*}
\begin{pmatrix}
X \\ 
W%
\end{pmatrix}
\sim
\mathcal N
\left(\begin{pmatrix}
 0 \\ 
 0
\end{pmatrix}
,\begin{pmatrix}
\sigma^2 & 0 \\ 
0 & 1%
\end{pmatrix}\right)
\end{equation*}
 where the variance $\sigma^2$ is varied in the experiments. Realizations of $A$ are generated independently of $(X,W)$ as follows:
  The random slope parameter $A_1$ is drawn either by a mixture of normal distributions, i.e., $\mathcal{N}\left( -1.5,1\right) $ and $\mathcal{N}\left( 1.5,0.5\right) $ with weights $0.5$, or by the  Gamma distribution $\Gamma(3,1)$. In each case, the random intercept is generated independently of the random slope by $A_0\sim\mathcal{N}( 0,1)$. 
Realizations of the dependent variable $Y$ are obtained by
\begin{equation}
Y= A_{0}+XB_{1},  \label{sim_lin}
\end{equation}%
with varying random coefficients
\begin{equation}
B_1=g_1(W)+A_1 \label{sim_quad}
\end{equation}
where in the experiments either $g_1(w)=\sin(w)$ or $g_1(w)=\exp(|w|)-1$.

In this simulation study, a linear sieve space with Hermite functions is used and no additional constraints are imposed. Consequently, the resulting estimator of the density of $B_1$ is of closed form as given in \eqref{VRS:estimator}. 
Numerical integration is used to compute the integrals in equation \eqref{VRS:estimator} based on the Adaptive Gauss-Kronrod quadrature (using the \textit{pracma} package in R). 
We keep the dimension for the random intercept fixed with $K_0=1$. 
The conditional characteristic function $h$ is estimated via series least squares as in equation \eqref{est:h} where the basis functions coincide with the Hermite functions of dimension $K_1$. 
For the estimation of the varying coefficient function $g_1$ we use quadratic B-spline bases with interior knots placed evenly. More precisely, we use two interior knots when using the function $g_1(w)=\sin(w)$ and four  knots in the case of $g_1(w)=\exp(|w|)-1$. 
\begin{table}[ht]
\renewcommand{\arraystretch}{1.25}
\centering
    {\small    
     \begin{tabular}{|c|c||ccc|}
        \hline
 Varying Coefficient & St. Dev. of $X$ & \multicolumn{3}{c|}{$\text{MISE}(\widehat f_{B_1})$ for sieve dim.}\\
           $g_1(w)$			&     $\sigma$ & $K_1=4$ & $K_1=5$ & $K_1=6$\\\hline\hline
            $\sin(w)$		& 1/2		& $\mathbf{0.0373}$	& 0.0642												& 0.3094\\
              &		$\sqrt{1/2}$   	& 0.0184											& $\mathbf{0.0167}$		& 0.0318\\
                								&		1   	& 0.0140 										& $\mathbf{0.0075}$	& 0.0092\\            
                  							&		2   	& 0.0129											& 0.0069												& $\mathbf{0.0056}$ \\\hline
        $\exp(|w|)-1$&1/2  		& $\mathbf{0.0267}$& 0.0672 												& 0.3110	\\
                	&$\sqrt{1/2}$   	& $\mathbf{0.0156}$	&0.0189													& 0.0336	\\
                								&		1   	& 0.0126											& $\mathbf{0.0088}$		& 0.0106 \\
                								&		2   	& 0.0125											& 0.0091												& $\mathbf{0.0070}$ \\
                \hline
        \end{tabular}
        }
       \caption{{\small  Monte Carlo Results for the $\text{MISE}(\widehat f_{B_1})$ for varying values of $\sigma$ and different functions of $g_1$. Bold letters show values of $\text{MISE}(\widehat f_{B_1})$ minimized w.r.t. $K_1$.}}
       \label{table:mise}
        \end{table}
        As weighting measure $\nu$, we choose the log-normal distribution as motivated in Example \ref{Exmp:LN}, more precisely, $\nu$ is given by $\textsl{Lognormal}(0,\sigma_\nu^2)$ where $\sigma_\nu=1/4$.

For the implementation of the estimator it is required to choose the dimension parameter $K_1$, given that we have fixed the dimension $K_0=1$.
In the Monte Carlo experiments we choose these parameters to minimize the mean integrated squared error.
Table \ref{table:mise} reports the mean integrated squared error of the estimator $\widehat f_{B_1}(\cdot,w)$ of $f_{B_1}(\cdot,w)$ when $w=0$ which is given by
\begin{align*}
\text{MISE}(\widehat f_{B_1})=\E \int_{-5}^5 \big|\widehat f_{B_1}(b,0) - f_{B_1}(b,0)\big|^2db
\end{align*} 
where the expectation is over the mean of all simulations. 
\begin{figure}[ht!]
	\centering
		\includegraphics[width=15cm]{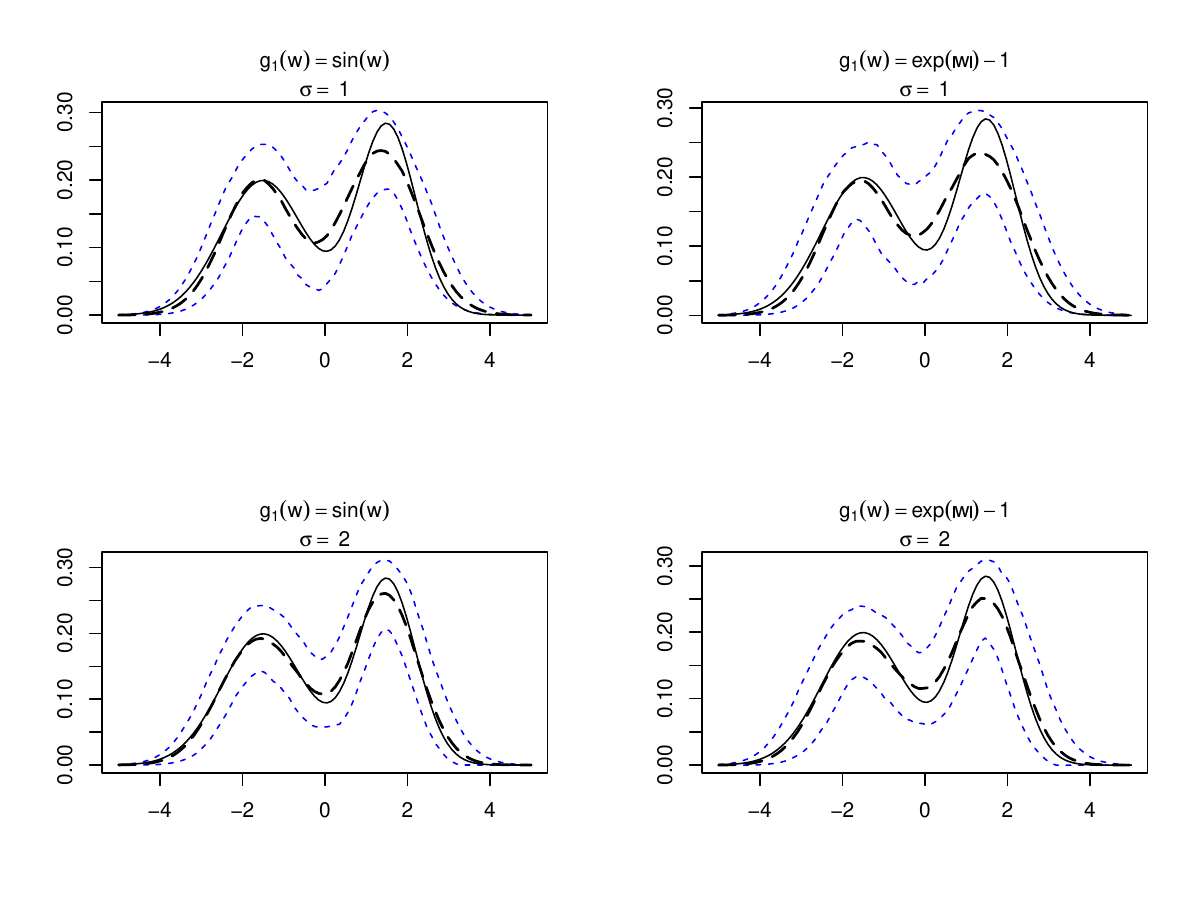}
		\vskip -1cm
	\caption{{\small The first column shows the median of the estimators $\widehat f_{B_1}(\cdot,w)$, with their pointwise 95\% confidence intervals when $g_1(w)=\sin(w)$. The second column is equivalent but uses $g_1(w)=\exp(|w|)-1$. The solid lines depict the true density $f_{B_1}(\cdot|w)$ where $w=0$. }}
	\label{dens-g1}
\end{figure}

Table \ref{table:mise} depicts values of $\text{MISE}(\widehat f_{B_1})$ for different values of $\sigma^2$ and  different functions $g_1$. In each case, the $\text{MISE}(\widehat f_{B_1})$ is provided for different sieve dimensions $K_1$, where the minimized value (w.r.t.  the sieve dimension $K_1$) is shown in bold letters. 
From Table \ref{table:mise} we see, not surprisingly, that the $\text{MISE}(\widehat f_{B_1})$ decreases as the variance of $X$ becomes larger. In particular, the optimal choice of the dimension parameter $K_1$ increases as the variance $\sigma^2$ becomes larger. This is in line with the rate of convergence as derived in Theorem \ref{thm:rate:fdb1}, see also the discussion thereafter.
The MISE when $g_1(w)=\exp(|w|)-1$ is larger as in the other case which is due to irregularity of the function $g_1$ at zero.

\begin{figure}[ht]
	\centering
		\includegraphics[width=15cm]{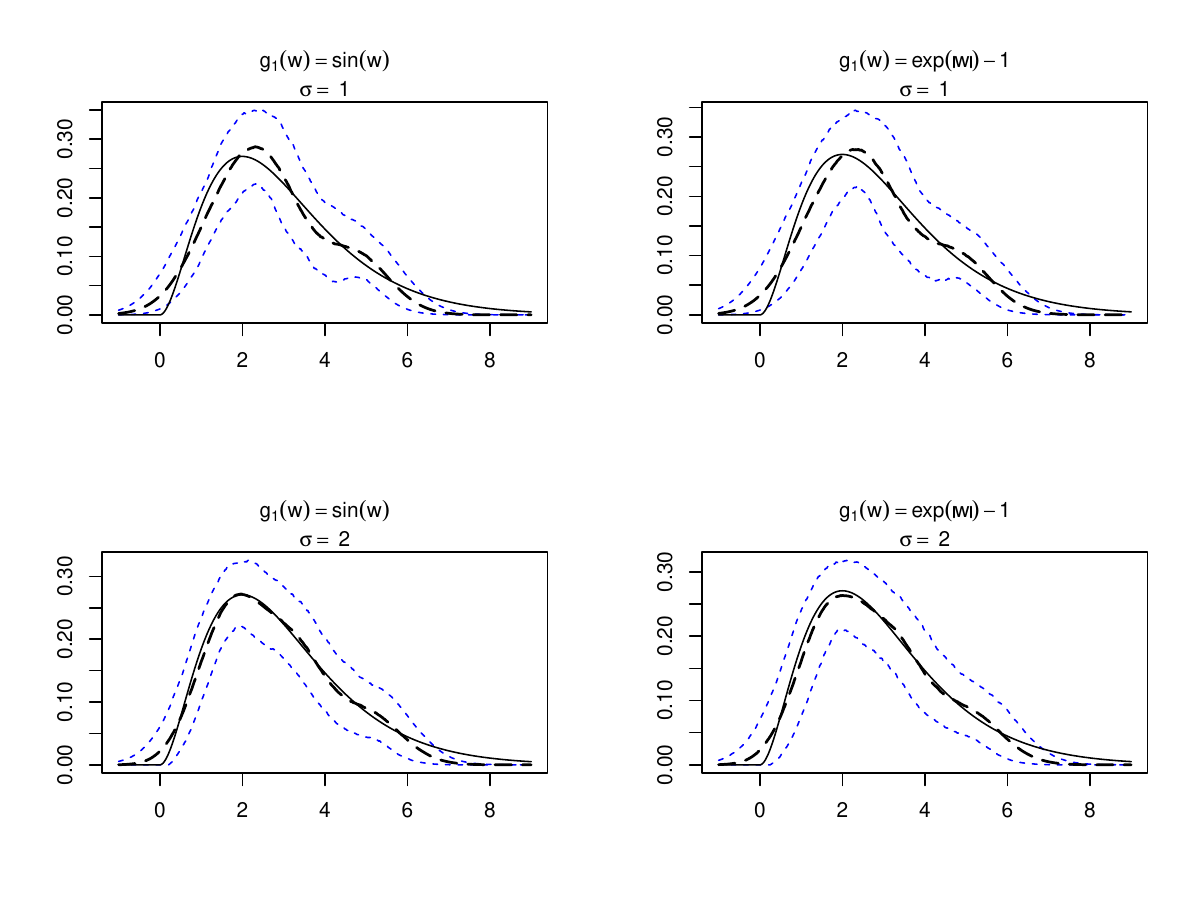}
		\vskip -1cm
	\caption{{\small As in Figure \ref{dens-g1} but  where the true density, depicted as solid line, is the density of $\Gamma(3,1)$.}}\label{dens-g2}
\end{figure}
Figure \ref{dens-g1} shows the estimation results when the true distribution of $B_1$ is given by an equally weighted mixture of $\mathcal{N}\left( -1.5,2\right) $ and $\mathcal{N}\left( 1.5,1\right) $. The solid line depicts the true density $f_{B_1}(\cdot,w)$ with $w=0$, the thick dashed line depicts the median of the estimators, and the thin dashed lines show the $95\%$ percent pointwise confidence intervals. 
The number of Hermite basis functions $K_1$ is chosen to minimize the  $\text{MISE}(\widehat f_{B_1})$ as shown in Table \ref{table:mise}, i.e., $K_1=5$ when $\sigma^2=1$ and $K_1=6$ when $\sigma^2=2$. We see that as the variance $\sigma^2$ increases (lower panel of the figure) the median of the estimators is closer to the true density. 
Although there is no positivity constraint imposed, the closed form median of the estimator together with their $95\%$ confidence intervals are non-negative between $-5$ and $5$. 
As the variance of $X$  increases the pointwise confidence intervals become more narrow, which is in  line with the pointwise asymptotic theory. This indicates the difficulty of estimating random coefficients in the case where $X$ has light tails. Nevertheless, we see that the estimator performs well even if $X$ has a small variance and is far from heavy tailed. Figure \ref{dens-g1} also shows that the procedure is robust even against irregularities of the varying coefficient function, i.e.,  when $g_1$ coincides with $g_1(w)=\exp(|w|)-1$.        
        
Figure \ref{dens-g2} depicts the estimator of $f_{B_1}$ in the case when $A_1$ is generated by the Gamma distribution $\Gamma(3,1)$. For the implementation of the estimator we use the same choice of tuning parameters as described above in the normal mixture case. Again we find the $95\%$ confidence interval and the median are more accurate when the variance of $X$ is increased from $1$ to $2$ and $g_1$ coincides with the sine function. In all cases, we see that the true density function lies outside of the confidence intervals for $b_1\in[7,9]$. This bias is due to a larger variance of $A_1$, i.e., $\Var(A_1)=3$, which implies that higher order Hermite functions are required to fully accurately capture the finite sample support of $A_1$. 

We finally present estimation results when not only the random slope (again we consider $A_1\sim\Gamma(3,1)$) but also  the random intercept is not normally distributed but $A_0\sim \mathcal U(0,1)$. We use the same implementation as above but normalize the VRS density estimator to integrate to one. Since we only use the Hermite basis function of order zero to account for the random slope the estimator is misspecified in this direction. The estimation results are shown in Figure \ref{dens-g3}. From this figure we see that misspecification of the density of $A_0$ has  only a minor effect on the accuracy of the VRS density estimator after normalization. This is in contrast to misspecification of the functional form of varying coefficient functions $g_l$ which can lead to severe nonlinear biases, see Example \ref{example:MC}. 

\begin{figure}[ht]
	\centering
		\includegraphics[width=15cm]{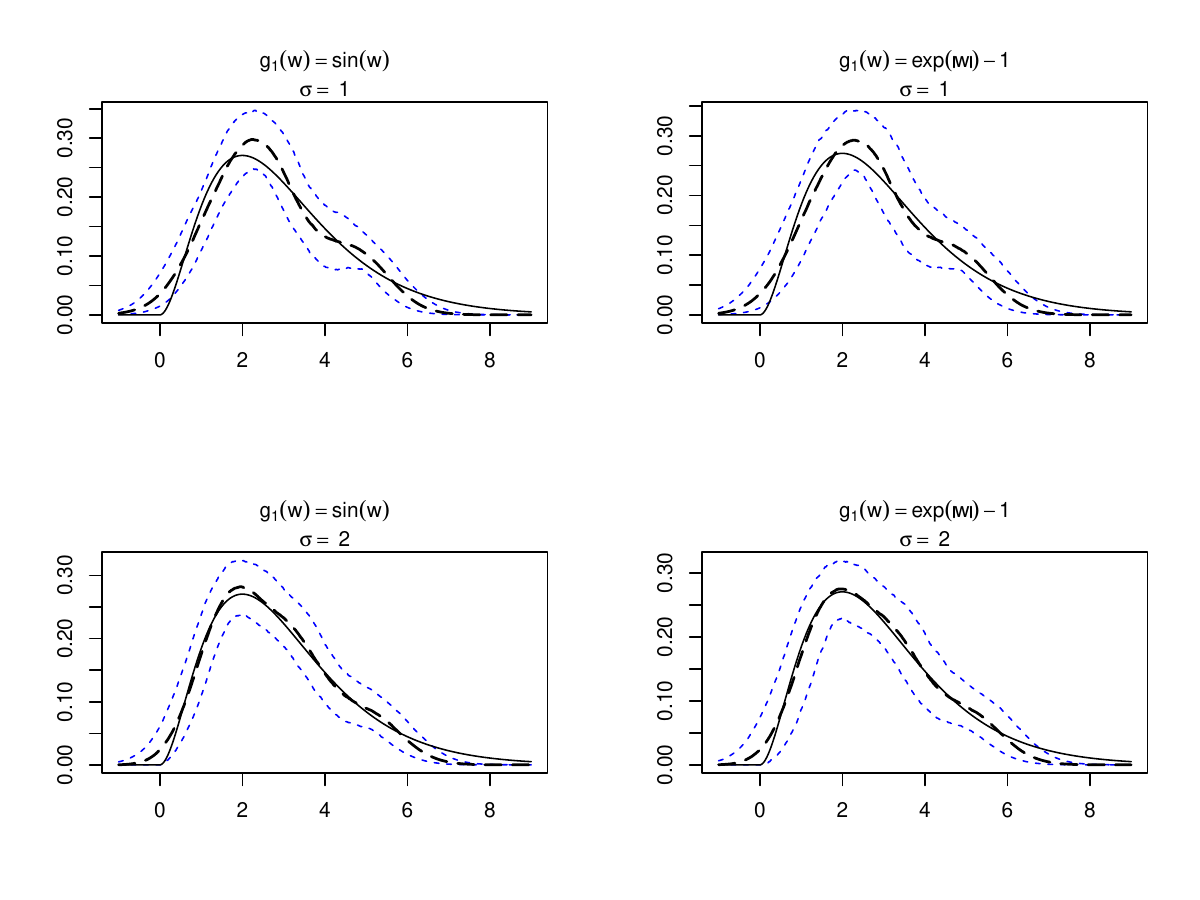}
		\vskip -1cm
	\caption{{\small As in Figure \ref{dens-g1} but  where the true density, depicted as solid line, is the density of $\Gamma(3,1)$ and $A_0\sim \mathcal U(0,1)$. Estimators are normalized to integrate to one.}}\label{dens-g3}
\end{figure}

\subsection{An Empirical Illustration}   
 \label{s_Application} 
In this subsection, the methodology is applied to analyze heterogeneity in income elasticity of demand for housing. Heterogeneity plays an important role in classical consumer demand and might be driven by unobserved heterogenous preferences. In the empirical illustration we  use data from the German Socio-Economic Panel (SOEP). While the SOEP is a longitudinal survey we restrict ourselves to the year 2013. We only consider individuals who do not have missing observations in rent, income, and size of the apartment, which results in a sample of size $n=7230$.

We are interested in assessing the heterogenous effect of household income on 
households' willingness to pay for rent. 
Formally, let us consider the empirical VRC model
\begin{align*}
Y=B_1X+g_0(W)+A_0, 
\end{align*}
and
\begin{align*}
B_1=g_1(W)+A_1,
\end{align*}
where $Y$ denotes the log monthly rent, $X$ denotes  the logarithm of the household net income per month, and $W$ is the logarithm of the size of the housing unit in square meters.\footnote{As stated in \cite{harrison1978}, rental prices reflect the market's current valuation of housing attributes, while housing values reflect expectations about future as well as present housing conditions. Hence, conceptually it is more appropriate to use rental prices when estimating hedonic functions for housing demand.}
The empirical VRC model thus imposes functional forms rather than letting the conditional distribution of unobserved heterogeneity given housing characteristics unrestricted. 
The following table
provides summary statistics of the relevant variables.
\begin{table}[ht]
\centering%
 {\small
\begin{tabular}{|c|ccccccc|}
\hline
& Min. & 1st Qu. & Median & Mean & 3rd Qu. & Max. & St. Dev. \\ \hline\hline
\text{$Y$: log rent} & 2.485    & 5.858   & 6.120   & 6.113  &  6.389  & 8.517 & 0.444 \\ 
\text{$X$: log hh. income}& 5.193 & 7.162 & 7.550 & 7.520 & 7.901 & 10.130  & 0.569\\
\text{$W$: log size housing}& 2.303 & 4.060 & 4.263 & 4.253 & 4.477 & 5.886 & 0.371
\\ \hline
\end{tabular}}
\end{table}

The interpretation of $B_1$ is that of a heterogeneous elasticity. Independence of the heterogeneous income  elasticity of demand and income itself might be difficult to justify if no additional covariates are included to explain $B_1$. We compute the variance of $B_1$ from  the empirical analog of $\E[(Y-g_1(W))^2 X]-\E[(Y-g_1(W))X]^2$ which yields the value of $0.0431$, where $g_1$ is estimated using B-splines as described below. The log size of housing $W$ explains much of the variation in  $B_1$, i.e., if $g_1\equiv 0$ then  variance of $B_1$ is given by $0.5899$. The small variance does not prevent estimating the density of $B_1$ using global basis functions, such as Hermite functions, because we can transform the model such that $B_1$ has , e.g., variance one and then back--transform the density function of $B_1$. 
 Here, $g_1$ is replaced by a B-spline estimator as explained below. 
\begin{figure}[ht]
	\centering
		\includegraphics[width=10cm]{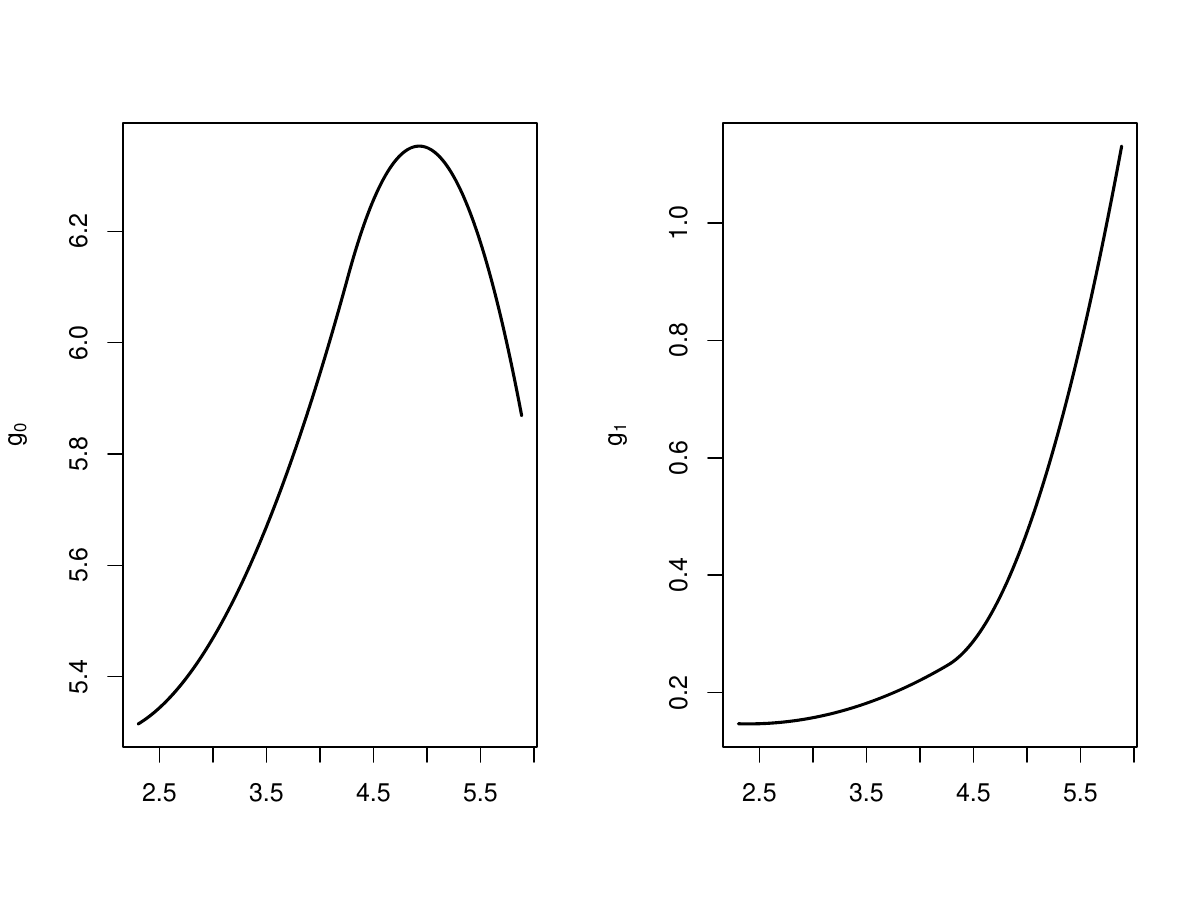}
		\vskip -1cm
	\caption{{\small Left: B-spline estimator of $g_0$. Right: B-spline estimator of $g_1$.  }}
	\label{g-est}
\end{figure}

The estimator $\widehat f_{B_1}$ is implemented as described in the previous section. The number of Hermite functions used is $K_0=1$ and $K_1=7$. The weighting measure $\nu$ is again given by $\textsl{Lognormal}(0,\sigma_\nu^2)$ with $\sigma_\nu=1/4$, as in the Monte Carlo section. For estimation of the functions $g_0$ and $g_1$ we use again quadratic B-spline bases functions with  three interior knots and follow Example \ref{exmp:series:vc}.
Figure \ref{g-est} depicts the B-spline estimators for the varying coefficient functions $g_0$ and $g_1$. We see that both estimators are nonlinear on the support of $W$. 

  For the bootstrap uniform confidence bands, we consider one representative sample and generate the bootstrap innovations $\varepsilon$ according to the two-point distribution suggested by \cite{mammen1993}, i.e., $\varepsilon$ equals $(1-\sqrt 5)/2$ with probability $(1+\sqrt 5)/(2\sqrt 5)$ and $(1+\sqrt 5)/2$ with probability $1-(1+\sqrt 5)/(2\sqrt 5)$.
Based on the estimator we generate the bootstrap process $\mathbb Z^*(\cdot|w)$ as described in Subsection \ref{sec:bootstrap:conf}. The results are based on $1000$ bootstrap iterations.

\begin{figure}[ht!]
	\centering
		\includegraphics[width=11cm]{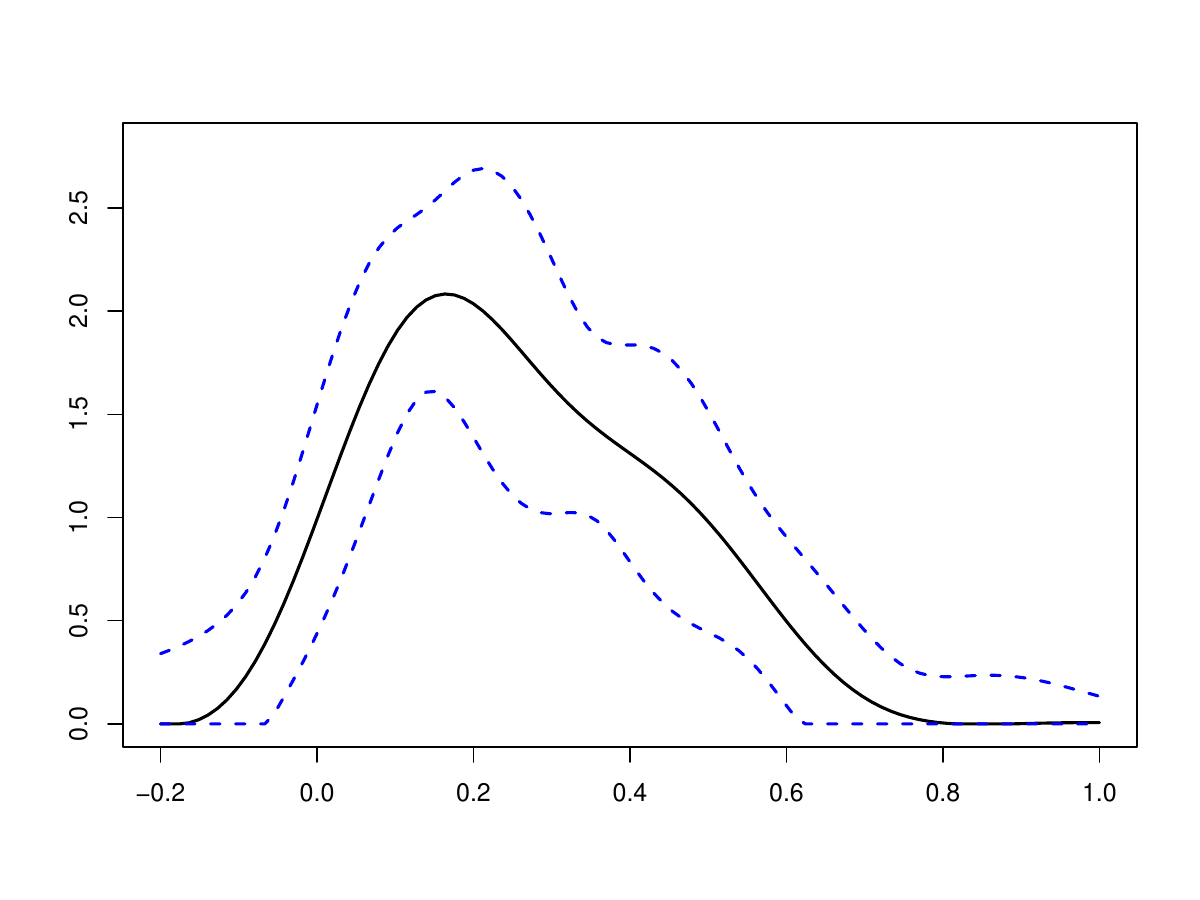}
		\vskip -1cm
	\caption{{\small Solid line depicts the sieve estimator of $f_{B_1}$ based on $K_1=7$ Hermite functions. Dotted lines depict the $95\%$  uniform confidence bands based on $1000$ bootstrap iterations. }}\label{housing-density}
\end{figure}
Figure \ref{housing-density} depicts the estimator for the density of $B_1$ evaluated at the mean of $W$ which is $w=4.253$. Note that $B_1$ can be directly interpreted as heterogenous marginal effect. 
From this figure we see that the estimated density has support between $-0.2$ and $0.8$. The uniform confidence bands show that the support is significantly positive (at $0.05$ nominal level) only at $-0.05$ and $0.6$. The estimated density is clearly not symmetric. 
We also see that the density is positively skewed and is more heavy tailed on the right hand side. This is reasonable as one would expect the response of a marginal increase of income to be skewed. 
 It is also interesting to see that the $95\%$ uniform confidence sets are bounded away from zero.

\section{Conclusion}\label{sec_Conclusion} 
This paper analyzes heterogeneity in  VRC models. This model generalizes ordinary RC models by including nonlinearities  in observed characteristics, which might stem, for instance, from measurement errors or control function residuals. 
A novel  estimator of the VRC density based on weighted sieve minimum distance is proposed. Under semiparametric restrictions on the random intercept, our estimator of the VRS density is not affected by the ill-posedness that is associated with the nonparametric estimation of the joint VRC density.
We establish novel inference results, such as uniform confidence bands,  to adress uncertainty in VRC density estimation which goes beyond what has been shown in ordinary RC models.  We find that finite sample estimation results are surprisingly stable when the sieve space is spanned by Hermite functions. This also advocates the use of the proposed methodology in the context of ordinary RC models. Finally, the methodology is applied to estimate the density of heterogeneous income elasticity of demand for housing, which is shown to be highly skewed. The proposed estimator can also be extended to include nonlinear index functions as in \cite{lewbel2017unobserved}. Yet the analysis of its asymptotic properties is left to future research.

\appendix
\section{Appendix}
Throughout the proofs, we will use $C > 0$ to denote a generic finite constant that may be different in different uses. 
Further, for ease of notation we write $\sum_j$ for $\sum_{j=1}^n$ and $\int$ for $\int_{\R^d}$ or $\int_{\R^{d-1}}$. Recall that $\|\cdot\|$ denotes the usual Euclidean norm, while for a matrix $A$, $\|A\|$ is the operator norm. 
 Recall the notation $P=\E[ p^K(X)p^K(X)']$ and $\widehat P= n^{-1}\sum_jp^K(X_j)p^K(X_j)'$. We use the notation $a_n\lesssim b_n$ to denote $a_n\leq C b_n$ for all $n\geq 1$.  
 
 \begin{proof}[\textsc{Proof of Lemma \ref{thm:ident:me}.}] 
The VRC model $(\ref{mod:gen}$--$\ref{mod:gen:rc})$ yields by Assumption \ref{ass:ind} (ii) the conditional moment restriction $\E[Y|X,W]= g_0(W)+\sum_{l=1}^{d-1} g_l(W)X_l$. 
 The varying coefficients functions $g_l$, $0\leq l\leq d-1$, are identified through this conditional moment restriction by  Assumption \ref{ass:ind} (iii) . 
Further, we obtain
\begin{align*}
\E[\exp(it(Y-g(S)))|X=x]=\E[\exp(it(A_0+A_1' X))|X=x].
\end{align*}
Since $X$ is independent of $A$ (see Assumption Assumption \ref{ass:ind} (i)) we can rewrite this equation using the notation of the Fourier transform for any $x$ in the support of $X$ as
\begin{align*}
h(x,t;g)=(\mathcal F f_A)(t, tx).
\end{align*} 
By the large support condition imposed on $X$ in Assumption \ref{ass:ind} (i) we can make use of Fourier inversion to obtain
\begin{align*}
f_A(a)&=\frac{1}{(2\pi)^{d}}\int \exp(-i a'u) (\mathcal F f_A)(u)du\\
&=\frac{1}{(2\pi)^{d}}\int \int |t|^{d-1}\exp\big(-i t(1,x')a\big) (\mathcal F f_A)(t,tx)dt\, dx\\
&=\frac{1}{(2\pi)^{d}} \int \int|t|^{d-1}\exp\big(-i t(1,x')a\big)h(x,t;g)dt\, dx,
\end{align*}
where the integral on the right hand side is finite due to Assumption \ref{ass:ind} (ii).
This shows identification of the RC density $f_A$ of $A$. Now identification of the VRC density of $B^w= \mathbf g(w)+A$ follows immediately by employing the relationship $f_{B}(b,w)= f_A(b-\mathbf g(w))$.
\end{proof}

\begin{proof}[\textsc{Proof of Lemma \ref{lem:vrs:ident}.}]
From the formula of the double series least squares estimator with $Q$ given in \eqref{def:mat:Q} and from basic properties of the Kronecker product we infer
\begin{align*}
Q&=(2\pi)^{d/2}\int\int\Big(\widetilde q^{K_0}(-t)\otimes \widetilde q^{K_1}(-tx)\Big)\Big(\widetilde q^{K_0}(t)\otimes \widetilde q^{K_1}(tx)\Big)'d\nu(t)dx\\
&=(2\pi)^{d/2}\int\int\widetilde q^{K_0}(-t)\widetilde q^{K_0}(t)'\otimes \widetilde q^{K_1}(-tx) \widetilde q^{K_1}(tx)'d\nu(t)dx\\
&=(2\pi)^{d/2}\int |t|^{1-d}\widetilde q^{K_0}(-t)\widetilde q^{K_0}(t)'d\nu(t)\otimes \int\widetilde q^{K_1}(-u) \widetilde q^{K_1}(u)'du\\
&=(2\pi)^{d/2}\int  |t|^{1-d} \widetilde q^{K_0}(-t) \widetilde q^{K_0}(t)'d\nu(t)\otimes \text{I}_{K_1}
\end{align*}
using that $q^{K_1}$ is a vector of Hermite functions which are orthonormal in $L^2(\mathbb R^{d-1})$. 
\end{proof}

\begin{proof}[\textsc{Proof of Proposition \ref{prop:eig}.}] Proof of (i). 
 Recall the definition $\widetilde\nu(t)=|t|^{1-d} \nu(t)$. For some constant $0<c<1$, for all $n\geq 1$, and any $a\in \mathbb R^K$ we have due to Parseval's Formula:
\begin{align*}
\| a\|^2
&= \int \int|a'\,q^K(t,u)|^2du \,dt\\
&= \int \int|a'\,(\mathcal F q^K)(t,u)|^2du \,dt\\
&=\int\int |a'\,(\mathcal F q^K)(t,tx)|^2\,|t|^{d-1} \1\{\widetilde\nu(t)\geq \tau_K\}dx\,dt \\
&\quad + \int \int|a'\,(\mathcal F q^K)(t,u)|^2 \1\{\widetilde\nu(t)< \tau_K\}du \,dt \\
&\leq  \tau_K^{-1} \int\int|a'\,(\mathcal F q^K)(t,tx)|^2\,|t|^{d-1}dx \, d\widetilde\nu (t) 
+  c \int \int|a'\,(\mathcal F q^K)(t,u)|^2du \,dt \\
&=  \tau_K^{-1} \int\int|a'\,(\mathcal F q^K)(t,tx)|^2dx\,d\nu (t) 
+  c\, \|a\|^2.
\end{align*}
Consequently, we obtain $\tau_K \text{I}_K\lesssim  Q$.\\

Proof of (ii). 
Using the series expansion of $f$ given by  $f=\sum_{k\geq 1} \langle f, q_k\rangle_{\R^d} q_k$ we obtain by the Cauchy-Schwarz inequality
\begin{align*}
 \int &\int\big|[\mathcal F(\Pi_K f-f)](t,tx)\big|^2dx\,d\nu(t)
 = \int \int \Big|\sum_{k\geq K} \langle f, q_k\rangle_{\R^d} (\mathcal F q_k)(t,tx)\Big|^2 dx\,d\nu(t)\\
&\leq \Big(\sum_{k\geq K} \langle f, q_k\rangle_{\R^d}^2\Big) \, \Big(\sum_{k\geq K}\int \int  |(\mathcal F q_k)(t,tx)|^2 dx\,d\nu(t)\Big)\\
&= \int(\Pi_K f-f)^2(b)\, db \,
\sum_{k\geq K}\int \int  |(\mathcal F q_k)(t,u)|^2 du\,d\widetilde\nu(t)\\
&= \int(\Pi_K f-f)^2(b)\, db \,
\sum_{k\geq K}\int  |(\mathcal F q_k)(t)|^2 d\widetilde\nu(t),
\end{align*}
by the unitary property of the Fourier transform, which completes the proof.
\end{proof}

For the following proofs we require additional notation. Introduce the vector 
\begin{align}\label{def:psi}
\psi^K(Y_j,S_j;\phi,t)=\big(\exp(it(Y_j-\phi(S_j))) - \exp(it (Y_j-g(S_j)))\big)\widetilde p^K(X_j)
\end{align}
with $k$--th entry denoted by $\psi_k(Y_j,S_j;\phi,t)$ and $\widetilde p^K(\cdot):=P^{-1/2}p^K(\cdot)$.
We also introduce the classes of function $\mathcal G=\{\phi=\beta'p_d^K:\|\phi-g\|_\infty\leq \sqrt{K/n}\}$,   $\mathcal F_k^R=\{\text{Re}(\psi_k(\cdot;\phi,\cdot)):\,\phi\in \mathcal G\}$,
and $\mathcal F_k^I=\{\text{Im}(\psi_k(\cdot;\phi,\cdot)):\,\phi\in \mathcal G\}$.
Further, $N_{[\,]}(\mathcal F,\|\cdot\|_{\nu,2}, \epsilon)$
denotes the $\|\cdot\|_{\nu,2}:=\sqrt{\int\|\cdot\|_{YS}^2\, d\nu} $ covering number with bracketing of a set of function $\mathcal F$. 
Define the envelope function $\Psi_k(\cdot)=\sup_{\phi\in\mathcal G}|\psi_k(\cdot;\phi,\cdot)|$, which satisfies
\begin{align}\label{bound:env}
\max_{1\leq k\leq K}\E\int|\Psi_k(Y, S;t)|^2 d\nu(t)&\leq \max_{1\leq k\leq K}\E\big[\sup_{\phi\in\mathcal G}\big|\big(\phi(S)-g(S)\big)\widetilde p_k(X)\big|^2\big]\int t^2d\nu(t)\nonumber \\
&\leq \sup_{\phi\in\mathcal G}\|\phi-g\|_\infty^2\max_{1\leq k\leq K}\E[\widetilde p_k^2(X)]\int t^2d\nu(t)\nonumber\\
&\lesssim K/n,
\end{align} 
using $\int t^2d\nu(t)\lesssim 1$ by Assumption \ref{Ass:bas} (ii). 
This upper bound is used in the following proofs. 
\begin{proof}[\textsc{Proof of Theorem \ref{thm:rate:fdb}.}]
The proof is based on the decomposition
\begin{align}\label{main:ineq}
\int\big|\widehat f_B(b,w)-f_B(b,w)\big|^2db
&\lesssim \int\big|\widehat f_A(a)-f_A(a)\big|^2da\nonumber\\
&\quad + \int \big|f_{A}\big(a-\widehat {\mathbf g}(w)\big)-f_{A}\big(a-\mathbf g(w)\big)\big|^2da\nonumber\\
&\lesssim \int\big|\widehat f_A(a)-\Pi_K f_A(a)\big|^2da+\int\big|\Pi_K f_A(a)-f_A(a)\big|^2da\nonumber\\
&\quad +\int \big|f_{A}\big(a-\widehat {\mathbf g}(w)\big)-f_{A}\big(a-\mathbf g(w)\big)\big|^2da.
\end{align}
Consider the first summand on the right hand side.  We have
\begin{align}\label{ineq:ill:posed}
\sup_{f\in\mathcal A_K}\set{\frac{\int f^2(b)db }{\int\int|( \mathcal{F} f)(t,tx)|^2d\nu(t)\,dx}}\lesssim \tau_K^{-1},
\end{align}
which is a consequence of the upper bounds imposed in Assumption \ref{Ass:bas}, that is, $\lambda_{\max}(\tau _K Q^{-1})\lesssim 1$ and $\lambda_{\max}\big(\int q^K(a)q^K(a)'da\big)\lesssim 1$, since for any $f(\cdot)=\beta'q^K(\cdot)$ it holds
\begin{align*}
\int\int|( \mathcal{F} f)(t,tx)|^2d\nu(t)\,dx
&=\beta'\, \int\int(\mathcal F q^K)(t,tx) (\mathcal F q^K)(t,tx)' dx\,d\nu (t)\, \beta\\
&\gtrsim\tau_K\|\beta\|^2\\
&\gtrsim\tau_K\int f^2(b)db.
\end{align*}
The definition of the estimator $\widehat f_A$ implies
\begin{multline*}
\int\int \Big|\widehat h(x,t;\widehat g) -(\mathcal{F}
\widehat f_A)(t, tx)\Big|^2d\nu(t)\,dx\\
\leq\int\int\Big|\widehat h(x,t;\widehat g) -(\mathcal{F}
\Pi_K f_A)(t, tx)\Big|^2d\nu(t)\,dx.
\end{multline*}
This inequality, the upper bound \eqref{ineq:ill:posed}, and the definition of the estimator $\widehat h$ yield
\begin{align*}
&\tau_K\int\big|\widehat f_A(a)-\Pi_K f_A(a)\big|^2da\lesssim
\int\int\Big|(\mathcal{F}
\widehat f_A)(t, tx) -(\mathcal{F}
\Pi_K f_A)(t, tx)\Big|^2d\nu(t)\,dx\\
&\lesssim \int\int\Big|\widehat h(x,t;\widehat g) -(\mathcal{F}
\Pi_K f_A)(t, tx)\Big|^2d\nu(t)\,dx\\
&\lesssim \underbrace{\int\int\Big|(\mathcal{F}f_A -\mathcal{F}
\Pi_K f_A)(t, tx)\Big|^2d\nu(t)\,dx}_{I}\\
&+\underbrace{\int\int\Big|p^K(x)'\widehat P^{-1}\frac{1}{n}\sum_j\exp\big(it(Y_j-g(S_j))\big)p^K(X_j)-p^K(x)'\gamma(t)\Big|^2d\nu(t)\,dx}_{II}\\
&+\underbrace{\int\int\Big|p^K(x)'\widehat P^{-1}\frac{1}{n}\sum_j\exp(it Y_j)p^K(X_j)\big(\exp\big(it \,\widehat g(S_j)\big) -\exp\big(it \, g(S_j)\big)\big)\Big|^2d\nu(t)\,dx}_{III}\\
&+\underbrace{\int\int\Big|p^K(x)'\gamma(t) - h(x,t;g)\Big|^2d\nu(t)\,dx}_{IV}.
\end{align*}
Due to the sieve approximation error of $f_A$ in Assumption \ref{Ass:lin} (i) it holds
\begin{align*}
I=\|\mathcal{F}
\Pi_K f_A-\mathcal{F}f_A\|_\nu^2=O\Big(\tau_K\|
\Pi_K f_A-f_A\|_{\R^d}^2\Big)=\tau_K K^{-2\zeta/d}.
\end{align*}
In the following, we make use of $\|\widehat P^{-1} - P^{-1}\| = O_p(\lambda_K^{-1}\sqrt{K\log(n)/n})$, see \cite[Lemma 6.2]{belloni2012} or \cite[Lemma 2.1]{chen2015}. 
By Assumption \ref{Ass:bas} (iii), the eigenvalues of $\int p^K(x)p^K(x)'dx$ are bounded from above and  thus
\begin{align*}
II
&\leq \Big\|\int p^K(x)p^K(x)'dx\Big\|\,\big\|\widehat P^{-1}P^{1/2}\big\|^2\\
&\quad\times\int\Big\|n^{-1}\sum_j\Big(\exp\big(it(Y_j-g(S_j))\big)-p^K(X_j)'\gamma(t)\Big)\widetilde p^K(X_j)\Big\|^2d\nu(t)\\
&\lesssim \lambda_K^{-1}\int\Big\|\E\Big[\big(h(X,t;g)-p^K(X)'\gamma(t)\big)\widetilde p^K(X)\Big]\Big\|^2d\nu(t)+O_p(K/(n\lambda_K))\\
&\lesssim \lambda_K^{-1}\int\E\big|h(X,t;g)-p^K(X)'\gamma(t)\big|^2d\nu(t)+O_p(K/(n\lambda_K))\\
&=O_p\big( \lambda_K^{-1}K^{-2\rho/(d-1)}+K/(n\lambda_K)\big),
\end{align*}
making use of the sieve approximation condition imposed on Assumption \ref{Ass:lin} (i). 
Consider $III$. Due to Assumption \ref{Ass:lin} (iii) we may assume $\widehat g\in\mathcal G$, which implies
\begin{align*}
III
&\lesssim \int\int\sup_{\phi\in \mathcal G}\Big|p^K(x)'\widehat P^{-1} P^{1/2} n^{-1}\sum_j\psi^K(Y_j,S_j;\phi,t)\Big|^2d\nu(t)dx
\end{align*}
by using the definition of $\psi^K$ as given in \eqref{def:psi}. 
Applying Theorem 2.14.5 of \cite{Vaart2000} together with the upper bound for the envelope function \eqref{bound:env} yields
\begin{align*}
&\sum_{k=1}^K\int\E\sup_{\phi\in \mathcal G}\Big|n^{-1}\sum_j\psi_k(Y_j,S_j;\phi,t)-\E\psi_k(Y,S;\phi,t)\Big|^2d\nu(t)\\
&\lesssim \frac{1}{n}\sum_{k=1}^K\Big(\int\E\sup_{\phi\in \mathcal G}\Big|n^{-1/2}\sum_j\psi_k(Y_j,S_j;\phi,t)-\E\psi_k(Y,S;\phi,t)\Big|d\nu(t)\sqrt{\frac{K}{n}}+\sqrt{\frac{K}{n}}\Big)^{2}\\
&\lesssim \frac{1}{n}\sum_{k=1}^K\Big(\int_0^1\sqrt{1+\log N_{[\,]}(\mathcal F_k^R, \|\cdot\|_{\nu,2}, \epsilon)}d\epsilon\,\sqrt{\frac{K}{n}}\\
&\qquad\qquad\qquad\qquad+\int_0^1\sqrt{1+\log N_{[\,]}(\mathcal F_k^I, \|\cdot\|_{\nu,2}, \epsilon)}d\epsilon\,\sqrt{\frac{K}{n}}+\sqrt{\frac{K}{n}}\Big)^{2}
\end{align*}
where the second upper bound is due to the last display of Theorem 2.14.2 of \cite{Vaart2000}.
The upper bound of the envelope function in inequality \eqref{bound:env} (implying local uniform $\|\cdot\|_{\nu,2}$ continuity of $\psi_k(\cdot;\phi,\cdot)$ with respect to $\phi\in\mathcal G$)
 together with Lemma 4.2 of \cite{Chen07} yields
\begin{align}\label{ineq:brack}
\log N_{ [\,]}(\mathcal F_k^R , \|\cdot\|_{\nu,2}, \epsilon)\leq \log N\big(\mathcal G, \|\cdot\|_S, \epsilon/4\big).
\end{align}
Using Assumption \ref{Ass:lin} (iv) we  thus obtain the rate $K/n$. 
Further, from the inequality $\lambda_{\max}(P^{-1/2}\int p^K(x)p^K(x)'dxP^{-1/2})\lesssim \lambda_K^{-1}$ we infer
\begin{align*}
\int\int&\sup_{\phi\in \mathcal G}\Big|p^K(x)'P^{-1/2}\E\psi^K(Y,S;\phi,t)\Big|^2d\nu(t)dx\\
&\lesssim\lambda_K^{-1}\int\sup_{\phi\in \mathcal G}\E\Big|\exp(it(Y-\phi(S))) - \exp(it (Y-g(S)))\Big|^2d\nu(t)\\
&\lesssim
\lambda_K^{-1} \sup_{\phi\in \mathcal G}\|\phi-g\|_S^2 \int t^2d\nu(t) \\
&=o_p\big(K/(n\lambda_K)\big),
\end{align*}
using $\E |\exp(it(Y-\phi(S))) - \exp(it (Y-g(S)))|^2\leq t^2\|\phi-g\|_S^2$ and $\int t^2d\nu(t)\lesssim 1$ by Assumption \ref{Ass:bas} (ii). 
Moreover, by Assumption \ref{Ass:lin} we have the sieve approximation bias
\begin{align*}
IV=\|h-\gamma'p^K\|_\nu^2
\lesssim K^{-2\rho/(d-1)}.
\end{align*}
In what follows, $Df$ denotes the Jacobian matrix of a function $f$. 
Finally, we consider 
\begin{align*}
\int \big|f_{A}&(b-\widehat {\mathbf g}(w))-f_{A}(b-{\mathbf g}(w))\big|^2db\\
&=\int\Big\|\int_0^1 D f_{A}\Big(b-u \,\widehat {\mathbf g}(w)+(u-1)\, \mathbf g(w)\Big)du\Big\|^2db\, \big\|\widehat{\mathbf g}(w)-\mathbf g(w)\big\|^2
\end{align*}
Continuity of $Df_A$ and consistency of $\widehat {\mathbf g}$ implies 
\begin{align*}
\int\Big\|\int_0^1 D f_{A}\Big(b-u \,\widehat {\mathbf g}(w)+(u-1)\, \mathbf g(w)\Big)du\Big\|^2db&=\int\|D f_{A}(a)\|^2da+o_p(1)\\
&\lesssim 1+o_p(1).
\end{align*}
 The result follows due to the rate restriction imposed in Assumption \ref{Ass:lin} (iii).
\end{proof}

\begin{proof}[\textsc{Proof of Theorem \ref{thm:rate:fdb1}.}]
We make use of the notation $\beta=\int q^{K}(a) f_{A_1}(a)da$ where  $q^K(t, u)=q^{K_0}(t)\otimes q^{K_1}(u)$.  In light of the main decomposition \eqref{main:ineq} in the proof of Theorem \ref{thm:rate:fdb} it is sufficient to consider
\begin{align*}
&\int|\widehat f_{A_1}(a)-f_{A_1}(a)|^2da\\
&\lesssim
\Big\|\int\int b_{K_0}(t)\widetilde q^{K_1}(-tx)\widehat
h(x,t; \widehat g) d\nu(t)\,dx- \int \beta'(q^{K_0}(a_0)\otimes\text{I}_{K_1})da_0\Big\|^2\\
&\quad +\int\Big|\int \beta'(q^{K_0}(a_0)\otimes q^{K_1}(a_1))-f_{A}(a_0,a_1)da_0\Big|^2da_1\\
&\lesssim
\Big\|Q^{-1}\int\int \big(\widetilde q^{K_0}(-t)\otimes\widetilde q^{K_1}(-tx)\big)\widehat
h(x,t; \widehat g) d\nu(t)\,dx-\beta\Big\|^2\\
&\quad +\int\big|\beta' (q^{K_0}(a_0)\otimes q^{K_1}(a_1))-f_{A}(a_0,a_1)\big|^2d(a_0,a_1),
\end{align*}
using the notation $b_{K_0}(t)=\int q^{K_0}(a)'Q_0^{-1}\widetilde q^{K_0}(-t) da$ and Lemma \ref{lem:vrs:ident}, that is, $Q^{-1}=Q_0^{-1}\otimes\text{I}_{K_1}$. 
Since $\int|\beta' q^{K}(a)-f_{A}(a)|^2da=O(K_1^{-2\zeta/(d-1)})$ due to condition $K_0=O(1)$, we only need to bound the first term on the right hand side. We further observe
\begin{align*}
\beta&=\int q^{K}(a) f_{A}(a)\,d(a)\\
&=Q^{-1}\int\int\widetilde q^{K}(-t,-tx)
(\mathcal F f_A)(t,tx)d\nu(t)\,dx\\
&=Q^{-1} \int\int\widetilde q^{K}(-t,-tx)
h(x,t; g) d\nu(t)\,dx.
\end{align*}
Therefore, it is sufficient to show
\begin{align*}
\Big\|Q_1^{-1} \int\int \widetilde q^{K}(-t,-tx)\big(\widehat h(x,t; \widehat g)-h(x,t; g)\big) d\nu(t)\,dx\Big\|^2=O_p\big(K_1/(n\lambda_{K_1})\big).
\end{align*}
This upper bound follows immediately from the proof of Theorem \ref{thm:rate:fdb} by using that $K_1$ is the dimension of basis functions used for the estimator $\widehat h$ and that $K_0=O(1)$,  which completes the proof. 
\end{proof}

\begin{proof}[\textsc{Proof of Theorem \ref{thm:inference}.}]
To simplify notation, let $\mathbf s:=Q^{-1/2}\ell\big(q^K(\cdot - \mathbf g(w))\big)$.
 Making use of Assumption \ref{A:inf} (i), we obtain the following lower bound for the sieve variance
  \begin{align*}
\textsl{v}_K(w)=\mathbf s'\,\Sigma\, \mathbf s\gtrsim \lambda_K^{-1}\|\mathbf s\|^2
 \end{align*}
which is used throughout this proof. 
 The proof is based on  the decomposition
 \begin{align*}
\ell\big(\widehat f_B(\cdot,w)\big)&-\ell\big(f_B(\cdot,w)\big)
=\underbrace{\ell\big(\widehat f_{A}(\cdot- {\widehat {\mathbf g}}(w))\big)-\ell\big(\widehat f_{A}(\cdot-{{\mathbf g}}(w))\big)}_{I}\\
    &+\underbrace{\mathbf s'Q^{-1/2}\int\int (\mathcal F q^K)(-t,-tx)
\Big(\widehat h(x,t; \widehat g)-\widehat h(x,t; g)\Big) d\nu(t) dx}_{II}\\
 &+ \underbrace{\mathbf s'Q^{-1/2}\int\int (\mathcal F q^K)(-t, -tx)\Big(\widehat
h(x,t; g)-p^K(x)'\gamma(t)\Big) d\nu(t)dx}_{III}\\
&+ \underbrace{\mathbf s'Q^{-1/2}\int\int (\mathcal F q^K)(-t, -tx)\Big(p^K(x)'\gamma(t)- h(x,t; g)\Big) d\nu(t)dx}_{IV}\\
  &+\underbrace{\mathbf s'Q^{-1/2}\int\int (\mathcal F q^K)(-t, -tx)
(\mathcal F f_A)(t,tx)  d\nu(t)dx-\ell\big(f_B(\cdot,w)\big)}_{V},
 \end{align*}
 where we evaluate each summand separately in the following.
By  Assumption \ref{A:inf} (iv) the basis functions $q_l$ are continuously differentiable and thus
\begin{align*}
\sqrt{n}\, I
&\leq \sqrt{n}\, \Big\|\ell\Big(\int_0^1 D\widehat f_{A}\big(\cdot-u \, {\widehat {\mathbf g}}(w)+(u-1)\,  {\mathbf g}(w)\big)du\Big)\Big\|\,\big\|{\widehat {\mathbf g}}(w)-{\mathbf g}(w)\big\|\\
&=o_p\Big(\sqrt{\textsl{v}_K(w)}\Big)
\end{align*}
using that $n \|\widehat{\mathbf g}(w)-\mathbf g(w)\|^2=o_p\big(\textsl{v}_K(w)\big)$, consistency of $\widehat f_A$, and Assumption \ref{Ass:lin} (ii). 
To bound $II$, make use of the definition of $\psi^K$ as given in \eqref{def:psi} to obtain
\begin{align*}
&\sqrt{n}\, II\\
&=\mathbf s'\int R(t) \widehat P^{-1}\frac{1}{\sqrt n}\sum_j\exp\big(itY_j)\Big(\exp\big(-it\widehat g(S_j)\big)-\exp\big(-it g(S_j)\big)\Big)p^K(X_j) d\nu(t)\\
 &\lesssim \int\sup_{\phi\in \mathcal G}\Big|\mathbf s' R(t) \widehat P^{-1} P^{1/2}\frac{1}{\sqrt n}\sum_j\psi^K(Y_j,S_j;\phi,t)\Big|d\nu(t)\\
 &\lesssim\Big(\int\Big\|\mathbf s' R(t)P^{-1/2}\Big\|^2d\nu(t)\Big)^{1/2}\\
 &\qquad\times\Big(\int\sup_{\phi\in \mathcal G}\Big\|\frac{1}{\sqrt n}\sum_j\Big(\psi^K(Y_j,S_j;\phi,t)-\E\psi^K(Y,S;\phi,t)\Big)\Big\|^2d\nu(t)\Big)^{1/2}\\
 &\quad +\sqrt{n}\int\sup_{\phi\in \mathcal G}\Big|\mathbf s' R(t)P^{-1/2} \E\psi^K(Y,S;\phi,t)\Big|d\nu(t)+o_p(1).
\end{align*}
Assumption $\int \| R(t)\|^2d\nu(t)=O(1)$ implies
\begin{align*}
\int\Big\|\mathbf s' R(t)P^{-1/2} \Big\|^2d\nu(t) &\leq \|\mathbf s\|^2 \int \| R(t)\|^2d\nu(t)\|P^{-1/2}\|^2\\
&\lesssim \textsl{v}_K(w).
\end{align*}
Making use of the upper bound \eqref{bound:env} for the envelope function of $\mathcal F^R_k$ and $\mathcal F^I_k$ and applying Theorem 2.14.5 of \cite{Vaart2000} yields
\begin{align*}
&\sum_{l=1}^K\int\E\sup_{\phi\in \mathcal G}\Big|n^{-1/2}\sum_j\psi_{l}(Y_j,S_j;\phi,t)-\E\psi_{l}(Y,S;\phi,t)\Big|^2d\nu(t)\\
&\lesssim \sum_{l=1}^K\Big(\int\E\sup_{\phi\in \mathcal G}\Big|n^{-1/2}\sum_j\psi_{l}(Y_j,S_j;\phi,t)-\E\psi_{l}(Y,S;\phi,t)\Big|d\nu(t)\sqrt{K/n}\\
&\qquad\qquad\qquad\qquad\qquad\qquad+\sqrt{K/n}\Big)^{2}\\
&\lesssim \frac{K}{n}\sum_{k=1}^K\Big(\int_0^1\sqrt{1+\log N_{[\,]}(\mathcal F_k^R, \|\cdot\|_{\nu,2}, \epsilon)}d\epsilon\\
&\qquad\qquad\qquad\qquad\qquad\qquad+\int_0^1\sqrt{1+\log N_{[\,]}(\mathcal F_k^I, \|\cdot\|_{\nu,2}, \epsilon)}d\epsilon+1\Big)^{2}\\
&\lesssim C_n^2 K^2/n
\end{align*}
 due to inequality \eqref{ineq:brack} and Assumption \ref{Ass:lin} (iv). Further, we have
\begin{align*}
\int\E\sup_{\phi\in \mathcal G}&\Big|\mathbf s'R(t)P^{-1/2}\E\psi^K(Y,S;\phi,t)\Big|d\nu(t)\\
&= \sup_{\phi\in \mathcal G}\|\phi-g\|_S \|\mathbf s\| \Big(\int \| R(t)\|^2d\nu(t)\Big)^{1/2}\|P^{-1/2}\|\\
&\lesssim  \sqrt{\textsl{v}_K(w)K/n}.
\end{align*}
Consequently, the rate restriction imposed on $K$ implies $\sqrt n II=o(\sqrt{\textsl{v}_K(w)})$. 
Consider $III$. Note that $\widehat P\gamma(t)=n^{-1}\sum_j h(X_j,t,g)p^K(X_j)$ and consequently we obtain by the definition of $R(t)$ that
\begin{align*}
\sqrt{n}\, III
&=\mathbf s'\int  R(t)\,  \widehat P^{-1}\frac{1}{\sqrt n}\sum_{j}\Big(\exp(it(Y_{j}-g(S_j)))p^K(X_{j}) -\widehat P\gamma(t)\Big)d\nu(t)\\
&\quad+o_p(1)\\
&=\mathbf s'\int  R(t)\,  P^{-1}\frac{1}{\sqrt n}\sum_{j}\rho_j(t)p^K(X_{j}) d\nu(t)+o_p(1),
\end{align*}
where $\rho_j(t)=\exp(it(Y_{j}-g(S_j)))-h(X_{j},t,g)$.
Consequently, we obtain
\begin{align*}
 \sqrt {n/\textsl{v}_K(w)}\,III
  &=\sum_j\underbrace{\big(n\, \textsl{v}_K(w)\big)^{-1/2} \mathbf s'  \int R(t)\rho_j(t)d\nu(t) P^{-1} p^K(X_{j})}_{\zeta_j}+o_p(1).
 \end{align*}
 We show $\sum_j\zeta_j\stackrel{d}{\rightarrow}\mathcal N(0,1)$ by the Lindeberg-Feller theorem. 
We see below that $\zeta_{j}$, $1\leq j \leq n$ satisfy Lindeberg's condition. It holds $\E[\zeta_{j}]=0$ and  $n\E[\zeta_{j}^2]=1$.
Using the lower bound for the sieve variance, for $\varepsilon>0$ we observe
\begin{align*}
 \sum_j& \E [\zeta_{j}^2\1\{|\zeta_{j}|>\varepsilon\}]
\leq n\varepsilon^2\E|\zeta_{j}/\varepsilon|^4\\
&\leq \textsl{v}_K(w)^{-2}n^{-1} \varepsilon^{-2}\|\mathbf s\|^4
\E\big\|\int R(t)\rho(t)d\nu(t) P^{-1} p^K(X)\big\|^4\\
&\lesssim n^{-1} \lambda_K^2\Big(\int\|R(t)\|^2d\nu(t)\Big)^2 \E\Big[\Big(\int|\rho(t)|^2d\nu(t)\Big)^2\big\|P^{-1} p^K(X)\big\|^2\Big]\\
&\qquad\times\sup_x\big\|P^{-1} p^K(x)\big\|^2\\
&\lesssim n^{-1} \lambda_K^{-1} K^2\\
&=o(1),
\end{align*}
by the rate condition $K^2=o(n\lambda_K)$ imposed in Assumption \ref{A:inf}. 
Consider $IV$. Using the notation $\check q^K(t, tx)= Q^{-1/2}(\mathcal F q^K)(t, tx)$ and linearity of the Fourier transform  we obtain
\begin{align*}
IV&=\mathbf s'Q^{-1/2}\int\int (\mathcal F q^K)(t, tx)\Big(p^K(x)'\gamma(t)- h(x,t; g)\Big) d\nu(t)dx\\
&=\sum_{l=1}^K \ell\Big((\mathcal F^{-1}\check q_l)(\cdot -\mathbf g(w))\Big)\langle\check q_l,\gamma'p^K- h\rangle_\nu \\
&\lesssim \ell\big([\mathcal F^{-1}(\gamma'p^K- h)](\cdot -\mathbf g(w))\big)\\
&=o\big(\sqrt{\textsl{v}_K(w)/n}\big),
\end{align*}
by Assumption \ref{A:inf} (iii). 
Consider $V$. By Assumption \ref{A:inf} (iv) the sieve space $\mathcal A_K$ is linear and thus, it holds
\begin{align*}
V&=\mathbf s'Q^{-1/2}\int\int(\mathcal F q^K)(t, tx)
(\mathcal F f_A)(t,tx) d\nu(t)\,dx-\ell\big(f_B(\cdot,w)\big)\\
&=\mathbf s'\underbrace{Q^{-1/2}\int \int (\mathcal F q^K)(t, tx)
[\mathcal F (f_A-\Pi_K f_A)](t,tx) d\nu(t)\,dx}_{=0}\\
&\quad+\ell\big(\Pi_K f_A(\cdot -\mathbf g(w))-f_A(\cdot -\mathbf g(w))\big)\\
&=o\big(\sqrt{\textsl{v}_K(w)/n}\big),
\end{align*}
by Assumption \ref{A:inf} (iii). Due to Lemma \ref{lem_cons_var}, equation \eqref{pw:cov:est},  the asymptotic distribution result remains valid as $\textsl{v}_K(w)$ is replaced by $\widehat{\textsl{v}}_K(w)$, which completes the proof.
\end{proof}

\begin{proof}[{\textsc{Proof of Theorem \ref{thm:bands}.}}]
Due to the \cite[Proof of Theorem 4.1]{ChenChristensen2015} it is sufficient to show
\begin{align*}
\left|\sqrt{n/\widehat{\textsl{v}}_K(b,w)}\big(  \widehat f_B(b,w) - f_B(b,w)\big) - \mathbb{Z}(b,w)\right|=o_p(r_n)
\end{align*}
since then the result follows by the anti-concentration inequality of \cite[Theorem 2.1]{chernozhukov2014}. 
Along the proof the inequality $  \textsl{v}_K(b,w)\gtrsim \lambda_K^{-1}\|Q^{-1/2}q^K(b-\mathbf g(w))\|^2$. Let $Z^n = \{(Y_1,X_1,W_1), \ldots, (Y_n,X_n,W_n)\}$.
We may assume that $\mathbf g(w)$ is known. Otherwise, consider some subset $\mathcal C'\subset\mathcal C$ where by employing consistency of  the vector valued function $\widehat{\mathbf g}$ we may assume that  $\set{b-\widehat{\mathbf g}(w): b\in\mathcal C}\subset \set{b-\mathbf g(w): b\in\mathcal C'}$. For simplicity of notation we assume in the following that $\mathbf g(w)=0$.
\\
\textbf{Step 1.} We start by showing that $\sqrt{n/\widehat{\textsl{v}}_K(b,w)}\big( \widehat f_B(b,w) - f_B(b,w)\big)$ can be uniformly approximated by the process
\begin{align*}
\widehat{\mathbb{Z}}(b,w) = \frac{q^K(b)' Q^{-1/2}}{\sqrt{\textsl{v}_K(b,w)}}\left(\frac{1}{\sqrt{n}}\sum_j\int  R(t) P^{-1} p^K(X_j) \rho_j(t)d\nu(t)\right),
\end{align*}
using the notation $\rho_j(t)=\exp(it(Y_j-g(S_j)))-h(X_j,t,g)$.
We observe
{\small \begin{align*}
  &\left|\sqrt{n/\widehat{\textsl{v}}_K(b,w)}\big(  \widehat f_B(b,w) - f_B(b,w)\big) - \widehat{\mathbb{Z}}(b,w)\right| \\
  &\leq 
  \underbrace{\left|\frac{\sqrt{n}q^K(b)'Q^{-1}}{\sqrt{ {\textsl{v}}_K(b,w)}}\int\int (\mathcal F q^K)(t, tx)\big(\widehat
h(x,t, g)-p^K(x)'\gamma(t)\big) d\nu(t)dx - \widehat{\mathbb{Z}}(b,w)\right|}_{I(b)}\\
&  + \left|\sqrt{\frac{ {\textsl{v}}_K(b,w)}{\widehat{\textsl{v}}_K(b,w)}} - 1\right| \Bigg(\underbrace{
\left|\frac{\sqrt{n}q^K(b)'Q^{-1}}{\sqrt{ {\textsl{v}}_K(b,w)}}\int\int (\mathcal F q^K)(t, tx)\big(\widehat
h(x,t, g)-p^K(x)'\gamma(t)\big) d\nu(t)dx \right|}_{II(b)}\\
&\qquad+  \underbrace{
\left|\frac{\sqrt{n}q^K (b)'Q^{-1}}{\sqrt{ {\textsl{v}}_K(b,w)}}\int\int (\mathcal F q^K)(t,tx)
\big(\widehat h(x,t, \widehat g)-\widehat h(x,t, g)\big) d\nu(t) dx\right|
 }_{III(b)}\\
 &\qquad+\underbrace{\left|\frac{\sqrt{n}q^K (b)'Q^{-1}}{\sqrt{ {\textsl{v}}_K(b,w)}}\int\int (\mathcal F q^K)(t, tx)\big(p^K(x)'\gamma(t) -  h(x,t, g)\big) d\nu(t)dx\right|}_{IV(b)}\\
  &\qquad+\underbrace{\left|\frac{\sqrt{n}}{\sqrt{ {\textsl{v}}_K(b,w)}}\Big(q^K (b)'Q^{-1}\int\int (\mathcal F q^K)(t, tx)h(x,t, g)d\nu(t)dx-f_B(b,w)\Big)\right|}_{V(b)}\Bigg).
\end{align*}}
We have $\| P^{-1/2}\widehat P P^{-1/2} - \text{I}_K\| = O_p(\sqrt{K\log(n)/(n\lambda_K)})$, see \cite[Lemma 2.1]{chen2015}. Further, we obtain
\begin{align*}
\sup_{b\in\mathcal{C}}I(b)&=\sup_{b\in\mathcal{C}}\left|\frac{q^K(b)'Q^{-1/2}}{\sqrt{\textsl{v}_K(b,w)}}\frac{1}{\sqrt{n}}\sum_j\int  R(t)\rho_j(t)d\nu(t) \big(\widehat P^{-1}-P^{-1}\big) p^K(X_j)\right|\\
&\leq \lambda_K\Big\|\frac{1}{\sqrt{n}}\sum_j \int  R(t)\rho_j(t)d\nu(t)\widehat P^{-1} P^{1/2}\big(P^{-1/2}\widehat P P^{-1/2}-\text{I}_K\big)\widetilde p^K(X_j)\Big\| \\
 &=O_p(\lambda_K^{-1} K\sqrt{\log(n)/n}).
\end{align*}
Define the process $\mathbb Z(b,w)=q^K(b)'Q^{-1/2}\mathcal Z/\sqrt{\textsl{v}_K(b,w)}$. We have
\begin{align*}
 \sup_{b\in\mathcal{C}} II(b) & \leq   \sup_{b\in\mathcal{C}}I(b)+ \sup_{b\in\mathcal{C}}|\widehat{\mathbb Z}(b,w)|\\
  & =  O_p\Big(\lambda_K^{-1}K\sqrt{\log(n)/n}\Big) + \sup_{b\in\mathcal{C}}\big|\widehat{\mathbb Z}(b,w) - \mathbb Z(b,w)\big| + \sup_{b\in\mathcal{C}}|\mathbb Z(b,w)|\\
  & =   O_p\Big(\lambda_K^{-1}K\sqrt{\log(n)/n}\Big) + o_p(r_n) + \sup_{b\in\mathcal{C}}|\mathbb Z(b,w)|\\
  & =  o_p(r_n) + O_p(c_n).
\end{align*}
where the third bound is due to step 2 below and the last equality is because of the condition $K^{5/2}=o(\lambda_K^2 r_n^3\sqrt{n})$ and by \cite[Lemma G.5]{ChenChristensen2015}, which is valid under our assumptions and which implies $\sup_{b\in\mathcal{C}}|\mathbb{Z}(b,w)| = O_p(c_n)$.
Consider  $III(b)$. Using the definition of $\psi^K$ as given in \eqref{def:psi} we obtain
\begin{align*}
  \sup_{b\in\mathcal{C}}&\frac{\sqrt{n}\left|III(b) \right|}{\sqrt{{\textsl{v}}_K(b,w)}}\\
&=\sup_{b\in\mathcal{C}}\frac{\|q^K (b)'Q^{-1/2}\|}{\sqrt{{\textsl{v}}_K(b,w)}}\int\sup_{\phi\in \mathcal G}\Big\|R(t) P^{-1} \frac{1}{\sqrt n}\sum_j\psi^K(Y_j,S_j;\phi,t)\Big\|d\nu(t)\\
&\quad+o_p(\lambda_K^{-1}K\sqrt{\log(n)/n})\\
&=o_p(\lambda_K^{-1}K\sqrt{\log(n)/n})
\end{align*}
following the proof of Theorem \ref{thm:inference}. Moreover, we observe
\begin{align*}
&\sup_{b\in\mathcal{C}}\frac{\sqrt{n}\left|IV(b) \right|}{\sqrt{ {\textsl{v}}_K(b,w)}}\\
 &\leq 
\sup_{b\in\mathcal{C}}\frac{\sqrt{n}\|q^K (b)'Q^{-1/2}\|}{\sqrt{ {\textsl{v}}_K(b,w)}} \Big\|Q^{-1/2}\int\int (\mathcal F q^K)(t, tx)\big(p^K(x)'\gamma(t) - h(x,t; g)\big) d\nu(t)dx\Big\|\\
&\lesssim \|\gamma'p^K-
h\|_\nu\\
&\lesssim  O(K^{-\rho/(d-1)})
\end{align*}
by Assumption \ref{Ass:lin} (i). 
 For the last summand we note
\begin{align*}
 \sup_{b\in\mathcal{C}}\frac{\sqrt{n}\left|V(b) \right|}{\sqrt{ {\textsl{v}}_K(b,w)}} &\leq \sup_{b\in\mathcal{C}}\frac{\sqrt{n}}{\sqrt{{\textsl{v}}_K(b,w)}}\big|\Pi_K f_B(b,w)-f_B(b,w)\big|.
\end{align*}
Consequently, Lemma \ref{lem_cons_var}, i.e.,  $\sup_{b\in\mathcal{C}}\left|\sqrt{{\textsl{v}}_K(b,w)/\widehat{\textsl{v}}_K(b,w)} - 1\right| = O_p(\sqrt{\zeta_n})$ and
the rate requirement in Assumption \ref{Ass:uniform} (iii) imply
\begin{align*}
\left|\sqrt{n/\widehat{\textsl{v}}_K(b,w)}\big(  \widehat f_B(b,w) - f_B(b,w)\big) - \widehat{\mathbb{Z}}(b,w)\right|=o_p(r_n).
\end{align*}

\noindent\textbf{Step 2.} 
Assumption $\int \| R(t)\|^2d\nu(t)=O(1)$ implies
\begin{align*}
\sum_j\E&\left\|\frac{1}{\sqrt n}\int  R(t) P^{-1} p^K(X_j) \rho_j(t) d\nu(t)\right\|^3\\
&\lesssim\frac{1}{\sqrt n}\Big(\int \| R(t)\|^2 d\nu(t)\Big)^{3/2}\, \E \| P^{-1} p^K(X)\|^3\\
&\lesssim \frac{K^{3/2}}{\sqrt n\lambda_K^{2}}.
\end{align*}
 Further,  recall that $r_n$ is a sequence satisfying
 \begin{align*}
 \frac{K^{5/2}}{\lambda_K^2 r_n^3\sqrt{n}}=o(1).
 \end{align*}
Hence we may apply Yurinskii's coupling (\cite[Theorem 10]{2002Pollard}) and consequently,  there exists a sequence of $\mathcal N(0,\Sigma)$ distributed random vectors $\mathcal Z$ such that
\begin{equation}\label{proof:Th:3:4:step:2_1}
  \left\|\frac{1}{\sqrt{n}}\int  R(t) P^{-1} p^K(X_j) \rho_j(t) d\nu(t) - \mathcal Z\right\| = o_p(r_n).
\end{equation}
Recall the definition $\mathbb Z(b,w)=q^K(b)'Q^{-1/2}\mathcal Z/\sqrt{\textsl{v}_K(b,w)}$, which is a centered Gaussian process with covariance function 
\begin{align*}
\E[\mathbb Z(b_1,w)\mathbb Z(b_2,w)] = q^K(b_1)'Q^{-1/2}\, \Sigma\, Q^{-1/2} q^K(b_2)\Big/\sqrt{ \textsl{v}_K(b_1,w)\textsl{v}_K(b_2,w)}.
\end{align*} 
Hence, by equation \eqref{proof:Th:3:4:step:2_1} we have
\begin{equation}\label{eq_step_2_UCB_exogenous}
  \sup_{b\in\mathcal{C}}\left|\widehat{\mathbb Z}(b,w) - \mathbb Z(b,w)\right| = o_p(r_n).
\end{equation}

\noindent\textbf{Step 3.} In this step we approximate the bootstrap process by a Gaussian process. Under the bootstrap distribution $\mathbb{P}^*$ each term $\int  R(t) P^{-1} p^K(X_j) \widehat\rho_j(t) d\nu(t)\varepsilon_j$ has mean zero for all $1\leq j\leq n$. Moreover, we have
\begin{equation*}
\frac{1}{n}\sum_j\E\left[\left.\int\int R(s) \widehat P^{-1} p^K(X_j)\widehat \rho_j(s)\varepsilon_j^2\widehat \rho_j(-t) p^K(X_j)'\widehat P^{-1}R(-t)'d\nu(s)d\nu(t)\right|Z^n\right]
= \widehat \Sigma.
\end{equation*}
Since $\E[|\varepsilon_j|^3|Z^n] < \infty$ uniformly in $j$, we have
\begin{align*}
\sum_j\E&\left[\left.\left\|\frac{1}{\sqrt{n}}\int  R(t) P^{-1} p^K(X_j) \widehat\rho_j(t) d\nu(t)\varepsilon_j\right\|^3\right|Z^n\right]\\
 &\lesssim \frac{1}{\sqrt{n}}\Big(\int \| R(t)\|^2 d\nu(t)\Big)^{3/2}\,\E\|P^{-1}p^K(X)\|^2\sup_x\|p^K(x)\| \\
&= O\left(\frac{K^{3/2}}{\sqrt{n}\lambda_K^2}\right).
\end{align*}
Again using \cite[Theorem 10]{2002Pollard}, conditional on the data $Z^n$, implies existence of  a  $\mathcal N(0,\widehat \Sigma)$ distributed random vectors $\mathcal Z^*$ such that
\begin{equation*}
  \left\|\frac{1}{\sqrt{n}}\sum_i \int  R(t) P^{-1} p^K(X_j) \widehat\rho_j(t) d\nu(t) - \mathcal Z^*\right\| = o_{p^*}(r_n)
\end{equation*}
wpa1. Therefore, 
\begin{align*}
  \sup_{a\in\mathcal{C}}\left|\mathbb{Z}^*(b,w) - \frac{q^K(b)'Q^{-1/2}}{\sqrt{\widehat{\textsl{v}}_K(b,w)}}\mathcal Z^*\right| = o_{p^*}(r_n)
\end{align*}
wpa1. Define a centered Gaussian process $\widetilde{\mathbb{Z}}(\cdot,w)$ under $\mathbb{P}^*$ as
\begin{align*}
  \widetilde{\mathbb Z}(b,w) = q^K(b)'Q^{-1/2}\Sigma^{1/2}\widehat \Sigma^{-1/2}\mathcal Z^*/\sqrt{ {\textsl{v}}_K(b,w)}
\end{align*}
which has the same covariance function as $\mathbb{Z}(b,w)$. By Lemma \ref{lem:gaussian:rate} below we have:
\begin{align*}
  \sup_{b\in\mathcal{C}}\left|\frac{q^K(b)'Q^{-1/2}}{\sqrt{\widehat{\textsl{v}}_K(b,w)}}\mathcal Z^* - \widetilde{\mathbb Z}(b,w)\right| = o_{p^*}(r_n)
\end{align*}
wpa1. This and the previous rate of convergence imply that
\begin{equation*}
  \sup_{b\in\mathcal{C}}\left|\mathbb Z^*(b,w) - \widetilde{\mathbb Z}(b,w)\right| = o_{p^*}(r_n)
\end{equation*}
wpa1, which completes the proof. 
\end{proof}

\section{Technical Assertions}

For the next result and the proof of it, recall the notation $\mathbf s=Q^{-1/2}\ell\big(q^K(\cdot - \mathbf g(w))\big)$ and let $\widehat{\mathbf s}:=Q^{-1/2} \ell\big(q^K(\cdot -\widehat{\mathbf g}(w))\big)$. 
\begin{lemma}\label{lem_cons_var}
 Let Assumptions \ref{A:inf}--\ref{Ass:uniform} be satisfied. Then,
\begin{align}
\Big|\sqrt{\frac{\widehat{\textsl{v}}_K(w)}{\textsl{v}_K(w)}} -1 \Big|&=o_p(1),\label{pw:cov:est}\\
\sup_{b\in\mathcal C}\Big|\sqrt{\frac{\widehat{\textsl{v}}_K(b,w)}{\textsl{v}_K(b,w)}} -1 \Big|&=O_p\Big(  \sqrt{(n\lambda_K\tau_K)^{-1/2}K^{1/2}
\log(n)+  K^{1/2-\rho/(d-1)}}\Big).\label{un:cov:est}
\end{align}
\end{lemma}
\begin{proof}
Proof of \eqref{pw:cov:est}. We make use of the decomposition 
\begin{align}\label{dec:cov:est}
\widehat{\textsl{v}}_K(w) &-\textsl{v}_K(w)= 
\mathbf s'\,\big(\widehat\Sigma-\Sigma\big)\mathbf s
+\big(\widehat{\mathbf s}-\mathbf s\big)'\,\widehat\Sigma\,\big(\widehat{\mathbf s}+\mathbf s\big).
\end{align}
We make use of the notation 
\begin{align*}
\widetilde \Sigma=\int_\R\int_\R  R(s) P^{-1} \frac{1}{n}\sum_{j=1}^n p^K(X_j)\rho_j(s)\rho_j(-t) p^K(X_j)'P^{-1}R(-t)'d\nu(s)d\nu(t)
\end{align*}
 and we may replace $\widehat P$ by $P$ in the definition of $\widehat \Sigma$. Also recall the definition $\widehat \rho_j(t)=\exp(it(Y_j-\widehat g(S_j)))-\widehat h(X_j,t,\widehat g)$. 
 We make use of of the following decomposition
 \begin{align*}
 &\widehat \rho_j(s)\overline{\widehat \rho_j(t)}-\rho_j(s)\overline{\rho_j(t)}\\
& = \big(\widehat \rho_j(s)-\rho_j(s)\big)\big(\overline{\widehat \rho_j(t)}-\overline{\rho_j(t)}\big)
 +\rho_j(s)\big(\overline{\widehat \rho_j(t)}-\overline{\rho_j(t)}\big)+\overline{\rho_j(t)}\big(\widehat \rho_j(s)-\rho_j(s)\big)
 \end{align*}
 and hence calculate
{\small \begin{align*}
&|\mathbf s'(\widehat\Sigma-\widetilde \Sigma)\mathbf s|\\
&\leq \Big|\mathbf s'\int\int  R(s) P^{-1}\frac{1}{n}\sum_jp^K(X_j)\Big(\widehat \rho_j(s)\overline{\widehat \rho_j(t)}-\rho_j(s)\overline{\rho_j(t)}\Big)p^K(X_j)' P^{-1}\overline{R(t)}' d\nu(s)d\nu(t)\mathbf s\Big|\\
&\leq \underbrace{n^{-1}\sum_j \Big|\mathbf s'\int  R(t) P^{-1}\big(\widehat \rho_j(t)-\rho_j(t)\big) d\nu(t) p^K(X_j) \Big|^2}_{I}\\
&+ 2\underbrace{\Big|\mathbf s'\int\int  R(s) P^{-1}\frac{1}{n}\sum_jp^K(X_j)\Big(\widehat \rho_j(s)-\rho_j(s)\Big)\overline{\rho_j(t)}p^K(X_j)' P^{-1}\overline{R(t)}' d\nu(s)d\nu(t)\mathbf s\Big|}_{II}.
\end{align*}}
For the first summand we evaluate using the definition of $\psi^K$ in equation \eqref{def:psi} and the Cauchy-Schwarz inequality that 
\begin{align*}
I&\lesssim
\int \E\big\|\mathbf s'R(t) P^{-1}p^K(X)\big\|^2d\nu(t)\,\|\widehat g- g\|_\infty^2\int t^2d\nu(t)\\
&+\int\Big\|\mathbf s'R(t) P^{-1/2}\Big\|^2d\nu(t)\\
&\quad\times\Big(\int\sup_{\phi\in \mathcal G}\Big\|n^{-1}\sum_j\psi^K(Y_j,S_j;\phi,t)-\E\psi^K(Y,S;\phi,t)\Big\|^2d\nu(t)\\
&\qquad\qquad+\int \sup_{\phi\in \mathcal G}\big\|\E\psi^K(Y,S;\phi,t)\big\|^2d\nu(t)\\
&\qquad\qquad+\sup_x\|p^K(x)\|^2\, n^{-1}\sum_j\int\big|p^K(X_j)'\gamma(t)- h(X_j,t, g)\big|^2d\nu(t)\Big)\\
&+o_p(K/(n\lambda_K))\\
&=O_p\Big(\|\mathbf s\|^2\lambda_K^{-1}\big(n^{-1}K+ n^{-1}K\log(n)
+ K^{1-2\rho/(d-1)}\big)\Big),
\end{align*}
following the  proof of Theorem \ref{thm:inference} and using that 
\begin{align*}
\int\E\big\|\mathbf s'R(t) P^{-1}p^K(X)\big\|^2d\nu(t)
&=\mathbf s'\int R(t)P^{-1}R(-t)d\nu(t) \,\mathbf s\\
&\leq\lambda_K^{-1} \|\mathbf s\|^2.
\end{align*}
Again following the  proof of Theorem \ref{thm:inference} and making use of the Cauchy-Schwarz inequality yields
\begin{align*}
II &\leq\sqrt{I}\times\sqrt{n^{-1}\sum_j \Big|\mathbf s'\int  R(t) P^{-1}\rho_j(t)d\nu(t) p^K(X_j) d\nu(t)\Big|^2}\\
&\leq \sqrt{I}\times\sqrt{n^{-1}\sum_j \int \Big|\mathbf s' R(t) P^{-1} p^K(X_j) \Big|^2d\nu(t)}\\
&\leq \sqrt{I}\times O_p\big(\lambda_K^{-1/2} \|\mathbf s\|\big)\\
&=O_p\Big( \|\mathbf s\|^2\lambda_K^{-1} (n^{-1/2}K^{1/2}\log(n)
+  K^{1/2-\rho/(d-1)})\Big)
\end{align*}
using the upper bound of $I$. 
Finally, we obtain
\begin{align*}
\mathbf s'(\widetilde\Sigma-\Sigma)\mathbf s=O_p\Big( \|\mathbf s\|^2\sqrt{K/(n \lambda_K)}\Big)
\end{align*}
which is due to the following calculation
\begin{align*}
&\E\|\widetilde\Sigma-\Sigma\|^2\\
&= \E\Big\|\int\int  R(s) P^{-1}\frac{1}{n}\sum_j\Big(p^K(X_j)\rho_j(s)\overline{\rho_j(t)}p^K(X_j)'- \E\big[p^K(X)\rho(s)\overline{\rho(t)}p^K(X)'\big]\Big)\\
&\qquad\times P^{-1}\overline{R(t)}' d\nu(s)d\nu(t)\Big\|^2\\
& \leq n^{-1}\E\Big\|\int R(t) P^{-1}p^K(X)\rho(t)d\nu(t)\Big\|^4\\
& \leq n^{-1}\int\big\| R(t) P^{-1/2}\big\|^4d\nu(t)\E\big\|P^{-1/2}p^K(X)\|^2\\
& \lesssim n^{-1}K/\lambda_K.
\end{align*}
For the second summand on the right hand side of \eqref{dec:cov:est} we observe
\begin{align*}
\big(\widehat{\mathbf s}-\mathbf s\big)'\,\widehat\Sigma\,\big(\widehat{\mathbf s}+\mathbf s\big)&
=\big(\widehat{\mathbf s}-\mathbf s\big)'\,\widehat\Sigma\,\big(\widehat{\mathbf s}-\mathbf s\big)
+2\big(\widehat{\mathbf s}-\mathbf s\big)'\,\widehat\Sigma\,\mathbf s
\end{align*}
In light of the upper bounds for $I$ and $II$ it is sufficient to bound $\big(\widehat{\mathbf s}-\mathbf s\big)'\,\Sigma\,\big(\widehat{\mathbf s}-\mathbf s\big)$  and $\big(\widehat{\mathbf s}-\mathbf s\big)'\,\Sigma\,\mathbf s$. 
Note that
\begin{align*}
\big(\widehat{\mathbf s}-\mathbf s\big)'\,&\Sigma\,\big(\widehat{\mathbf s}-\mathbf s\big)\leq\|\widehat{\mathbf s}-\mathbf s\|^2\\
&\leq \tau_K^{-1}\|\ell\big(q^K(\cdot -\widehat{\mathbf g}(w))-q^K(\cdot -\mathbf g(w))\big)\|^2\\
&\leq \tau_K^{-1}\Big\|\ell\Big(\int_0^1Dq^K(\cdot -u\widehat{\mathbf g}(w)+(u-1)\mathbf g(w))du\Big)\Big\|^2 \big\|\widehat{\mathbf g}(w)-\mathbf g(w)\big\|^2\\
&=O_p\big(K/(n\tau_K)\big).
\end{align*}
and similarly,
\begin{align*}
\big(\widehat{\mathbf s}-\mathbf s\big)'\,\Sigma\,\mathbf s&\leq\|\widehat{\mathbf s}-\mathbf s\| \|\mathbf s\|\\
&=O_p\big(\|\mathbf s\|\sqrt{K/(n\tau_K)}\big).
\end{align*}
Consequently, the previous inequalities together with the bound $\|\mathbf s\|^2/\textsl{v}_K(w)\leq \lambda_K$ (due to Assumption \ref{A:inf} (i)) and $\lambda_K^{-1/2}\lesssim \sqrt{\textsl{v}_K(w)}$ imply
\begin{align*}
\frac{\widehat{\textsl{v}}_K(w)}{\textsl{v}_K(w)} -1 
=O_p\Big(n^{-1/2}K^{1/2}\log(n)
+  K^{1/2-\rho/(d-1)}+\sqrt{K/(n\tau_K\lambda_K)}\Big),
\end{align*}
which, due to the rate condition imposed in Assumption \ref{A:inf:var} implies bound \eqref{pw:cov:est}. 
The result \eqref{un:cov:est} follows analogously.
\end{proof}

\begin{lemma}\label{lem:gaussian:rate}
  Let Assumptions \ref{A:inf}--\ref{Ass:uniform} be satisfied.
  Then,
  \begin{align*}
   \sup_{b\in\mathcal{C}}\left|\frac{q^K(b)'Q^{-1/2}}{\sqrt{\widehat{\textsl{v}}_K(b,w)}}\mathcal Z^* - \widetilde{\mathbb Z}(b,w)\right|  = o_{p^*}(r_n)
  \end{align*}
  with probability approaching one.
\end{lemma}
\begin{proof}
  The proof of this lemma follows the proof of \cite[Lemma G.6]{ChenChristensen2015} and so we provide only the main parts where the two proofs differ. We make the decomposition
  \begin{align*}
    \sup_{b\in\mathcal{C}}&\left|\frac{q^K(b)'Q^{-1/2}}{\sqrt{\widehat{\textsl{v}}_K(b,w)}}\mathcal Z^* - \widetilde{\mathbb Z}(b,w)\right| \\
   & \leq \underbrace{\sup_{b\in\mathcal{C}}\left|\frac{q^K(b)'Q^{-1/2}\left(\text{I}_K - \Sigma^{1/2}\widehat \Sigma^{-1/2}\right)}{\sqrt{ {\textsl{v}}_K(b,w)}}\mathcal Z^*\right|\sup_{b\in\mathcal{C}}\sqrt{\frac{\textsl{v}_K(b,w)}{\widehat{\textsl{v}}_K(b,w)}}}_{I}\\
     &\quad + \underbrace{\sup_{b\in\mathcal{C}}\left|\sqrt{\frac{{\textsl{v}}_K(b,w)}{\widehat{\textsl{v}}_K(b,w)}} - 1\right| \sup_{b\in\mathcal{C}}\left|\frac{q^K(b)'Q^{-1/2}}{\sqrt{{\textsl{v}}_K(b,w)}}\Sigma^{1/2}\widehat \Sigma^{-1/2}\mathcal Z^*\right|}_{II}.
  \end{align*}
 Let $\widetilde \Delta_w$ denote the standard deviation semimetric on $\mathcal{C}$ associated with the Gaussian process (under $\mathbb{P}^*$) ${\textsl{v}}_K(b,w)^{-1/2}q^K(b)'Q^{-1/2}\big(\text{I}_K - \Sigma^{1/2}\widehat \Sigma^{-1/2}\big)\mathcal Z^*$ and defined as
 \begin{align*}
 \widetilde \Delta_w(b_1,b_2)^2 = \E^*\left[\left(\left(\frac{q^K(b_1)}{\sqrt{ \textsl{v}_K(b_1,w)}} - \frac{q^K(b_2)}{\sqrt{\textsl{v}_K(b_2,w)}}\right)'Q^{-1/2}\big(\text{I}_K  - \Sigma^{1/2}\widehat \Sigma^{-1/2}\big)\mathcal Z^*\right)^2\right].
  \end{align*}
  We observe $\widetilde \Delta_w(b_1,b_2) \leq \Delta_w(b_1,b_2)\|\Sigma^{-1/2}\widehat \Sigma^{1/2} - \text{I}_K\|$ and
  \begin{align*}
    \|\Sigma^{-1/2}\widehat \Sigma^{1/2} - \text{I}_K \| &\leq \|\widehat{\Sigma}^{1/2} - \Sigma^{1/2}\|\,  \|{\Sigma}^{-1/2}\|\\
    &\leq (\lambda_{\min}^{1/2}(\Sigma) + \lambda_{\min}^{1/2}(\widehat{\Sigma}))^{-1}\|\widehat{\Sigma} - \Sigma\|\\
   & = O_p(\sqrt{\eta_n}).
   \end{align*}
where the last bound is due to the proof of Lemma \ref{lem_cons_var}.
  Thus, following line by line the proof of \cite[Lemma G.6]{ChenChristensen2015} we obtain that $I  = o_{p^*}(r_n)$ under Assumption \ref{Ass:uniform}.
  Next, let us consider term $II$ which is the supremum of a Gaussian process with the same distribution (under $\mathbb{P}^*$) as $\mathbb{Z}(b,w)$. Therefore, by applying Lemma \ref{lem_cons_var} and \cite[Lemma G.5]{ChenChristensen2015} we obtain $II  = o_{p^*}(r_n)$.
\end{proof}

\bibliography{BiB}

\begin{thebibliography}{40}
\providecommand{\natexlab}[1]{#1}
\providecommand{\url}[1]{\texttt{#1}}
\expandafter\ifx\csname urlstyle\endcsname\relax
  \providecommand{\doi}[1]{doi: #1}\else
  \providecommand{\doi}{doi: \begingroup \urlstyle{rm}\Url}\fi

\bibitem[Ai and Chen(2003)]{AC03econometrica}
C.~Ai and X.~Chen.
\newblock Efficient estimation of models with conditional moment restrictions
  containing unknown functions.
\newblock \emph{Econometrica}, 71:\penalty0 1795--1843, 2003.

\bibitem[Banks et~al.(1997)Banks, Blundell, and Lewbel]{banks1997}
J.~Banks, R.~Blundell, and A.~Lewbel.
\newblock Quadratic engel curves and consumer demand.
\newblock \emph{Review of Economics and Statistics}, 79\penalty0 (4):\penalty0
  527--539, 1997.

\bibitem[Belloni et~al.(2015)Belloni, Chernozhukov, Chetverikov, and
  Kato]{belloni2012}
A.~Belloni, V.~Chernozhukov, D.~Chetverikov, and K.~Kato.
\newblock Some new asymptotic theory for least squares series: Pointwise and
  uniform results.
\newblock \emph{Journal of Econometrics}, 186\penalty0 (2):\penalty0 345--366,
  2015.

\bibitem[Ben-Moshe et~al.(2017)Ben-Moshe, D'Haultf{\oe}uille, and
  Lewbel]{ben2017identification}
D.~Ben-Moshe, X.~D'Haultf{\oe}uille, and A.~Lewbel.
\newblock Identification of additive and polynomial models of mismeasured
  regressors without instruments.
\newblock \emph{Journal of Econometrics}, 200\penalty0 (2):\penalty0 207--222,
  2017.

\bibitem[Beran(1993)]{beran1993}
R.~Beran.
\newblock Semiparametric random coefficient regression models.
\newblock \emph{Annals of the Institute of Statistical Mathematics},
  45\penalty0 (4):\penalty0 639--654, 1993.

\bibitem[Beran and Hall(1992)]{beran1992}
R.~Beran and P.~Hall.
\newblock Estimating coefficient distributions in random coefficient
  regressions.
\newblock \emph{The Annals of Statistics}, pages 1970--1984, 1992.

\bibitem[Beran et~al.(1996)Beran, Feuerverger, and Hall]{beran1996}
R.~Beran, A.~Feuerverger, and P.~Hall.
\newblock On nonparametric estimation of intercept and slope distributions in
  random coefficient regression.
\newblock \emph{The Annals of Statistics}, 24\penalty0 (6):\penalty0
  2569--2592, 1996.

\bibitem[Blundell et~al.(2007)Blundell, Chen, and
  Kristensen]{BCK07econometrica}
R.~Blundell, X.~Chen, and D.~Kristensen.
\newblock Semi-nonparametric iv estimation of shape-invariant engel curves.
\newblock \emph{Econometrica}, 75\penalty0 (6):\penalty0 1613--1669, 2007.

\bibitem[Bongioanni and Torrea(2009)]{bongioanni2009}
B.~Bongioanni and J.~L. Torrea.
\newblock What is a sobolev space for the laguerre function systems.
\newblock \emph{Studia Math}, 192\penalty0 (2):\penalty0 147--172, 2009.

\bibitem[Breunig and Hoderlein(2018)]{breunig2016}
C.~Breunig and S.~Hoderlein.
\newblock Specification testing in random coefficient models.
\newblock \emph{Quantitative Economics}, 9\penalty0 (3):\penalty0 1371--1417,
  2018.

\bibitem[Chen(2007)]{Chen07}
X.~Chen.
\newblock Large sample sieve estimation of semi-nonparametric models.
\newblock \emph{Handbook of Econometrics}, 2007.

\bibitem[Chen and Christensen(2015)]{chen2015}
X.~Chen and T.~M. Christensen.
\newblock Optimal uniform convergence rates and asymptotic normality for series
  estimators under weak dependence and weak conditions.
\newblock \emph{Journal of Econometrics}, 2015.

\bibitem[Chen and Christensen(2018)]{ChenChristensen2015}
X.~Chen and T.~M. Christensen.
\newblock Optimal sup-norm rates and uniform inference on nonlinear functionals
  of nonparametric iv regression.
\newblock \emph{Quantitative Economics}, 9\penalty0 (1):\penalty0 39--84, 2018.

\bibitem[Chen and Pouzo(2012)]{Chen08}
X.~Chen and D.~Pouzo.
\newblock Estimation of nonparametric conditional moment models with possibly
  nonsmooth generalized residuals.
\newblock \emph{Econometrica}, 80\penalty0 (1):\penalty0 277--321, 01 2012.

\bibitem[Chen and Pouzo(2015)]{chen2013}
X.~Chen and D.~Pouzo.
\newblock Sieve quasi likelihood ratio inference on semi/nonparametric
  conditional moment models.
\newblock \emph{Econometrica}, 83\penalty0 (3):\penalty0 1013--1079, 2015.

\bibitem[Chernozhukov et~al.(2014)Chernozhukov, Chetverikov, and
  Kato]{chernozhukov2014}
V.~Chernozhukov, D.~Chetverikov, and K.~Kato.
\newblock Anti-concentration and honest, adaptive confidence bands.
\newblock \emph{The Annals of Statistics}, 42\penalty0 (5):\penalty0
  1787--1818, 2014.

\bibitem[Coppejans and Gallant(2002)]{coppejans2002}
M.~Coppejans and A.~R. Gallant.
\newblock Cross-validated snp density estimates.
\newblock \emph{Journal of Econometrics}, 110\penalty0 (1):\penalty0 27--65,
  2002.

\bibitem[Dunker et~al.(2019)Dunker, Eckle, Proksch, and
  Schmidt-Hieber]{dunker2017}
F.~Dunker, K.~Eckle, K.~Proksch, and J.~Schmidt-Hieber.
\newblock Tests for qualitative features in the random coefficients model.
\newblock \emph{Electronic Journal of Statistics}, 13\penalty0 (2):\penalty0
  2257--2306, 2019.

\bibitem[Fan et~al.(2003)Fan, Yao, and Cai]{fan2003}
J.~Fan, Q.~Yao, and Z.~Cai.
\newblock Adaptive varying-coefficient linear models.
\newblock \emph{Journal of the Royal Statistical Society: Series B (Statistical
  Methodology)}, 65\penalty0 (1):\penalty0 57--80, 2003.

\bibitem[Fox et~al.(2011)Fox, Ryan, and Bajari]{fox2011simple}
J.~T. Fox, S.~P. Ryan, and P.~Bajari.
\newblock A simple estimator for the distribution of random coefficients.
\newblock \emph{Quantitative Economics}, 2\penalty0 (3):\penalty0 381--418,
  2011.

\bibitem[Fox et~al.(2016)Fox, il~Kim, and Yang]{fox2016simple}
J.~T. Fox, K.~il~Kim, and C.~Yang.
\newblock A simple nonparametric approach to estimating the distribution of
  random coefficients in structural models.
\newblock \emph{Journal of Econometrics}, 195\penalty0 (2):\penalty0 236--254,
  2016.

\bibitem[Gautier and Le~Pennec(2018)]{gautier2011}
E.~Gautier and E.~Le~Pennec.
\newblock Adaptive estimation in the nonparametric random coefficients binary
  choice model by needlet thresholding.
\newblock \emph{Electronic Journal of Statistics}, 12\penalty0 (1):\penalty0
  277--320, 2018.

\bibitem[Harrison and Rubinfeld(1978)]{harrison1978}
D.~Harrison and D.~L. Rubinfeld.
\newblock Hedonic housing prices and the demand for clean air.
\newblock \emph{Journal of environmental economics and management}, 5\penalty0
  (1):\penalty0 81--102, 1978.

\bibitem[Hausman et~al.(1991)Hausman, Newey, Ichimura, and Powell]{hausman1991}
J.~A. Hausman, W.~K. Newey, H.~Ichimura, and J.~L. Powell.
\newblock Identification and estimation of polynomial errors-in-variables
  models.
\newblock \emph{Journal of Econometrics}, 50\penalty0 (3):\penalty0 273--295,
  1991.

\bibitem[Hoderlein et~al.(2010)Hoderlein, Klemel{\"a}, and
  Mammen]{hoderlein2010}
S.~Hoderlein, J.~Klemel{\"a}, and E.~Mammen.
\newblock Analyzing the random coefficient model nonparametrically.
\newblock \emph{Econometric Theory}, 26\penalty0 (03):\penalty0 804--837, 2010.

\bibitem[Hoderlein et~al.(2017)Hoderlein, Holzmann, and Meister]{hoderlein2014}
S.~Hoderlein, H.~Holzmann, and A.~Meister.
\newblock The triangular model with random coefficients.
\newblock \emph{Journal of econometrics}, 201\penalty0 (1):\penalty0 144--169,
  2017.

\bibitem[Hohmann and Holzmann(2016)]{hohmann2016}
D.~Hohmann and H.~Holzmann.
\newblock Weighted angle radon transform: Convergence rates and efficient
  estimation.
\newblock \emph{Statistica Sinica}, pages 157--175, 2016.

\bibitem[Imbens and Newey(2009)]{imbens2009}
G.~W. Imbens and W.~K. Newey.
\newblock Identification and estimation of triangular simultaneous equations
  models without additivity.
\newblock \emph{Econometrica}, 77\penalty0 (5):\penalty0 1481--1512, 2009.

\bibitem[Lewbel and Pendakur(2017)]{lewbel2017unobserved}
A.~Lewbel and K.~Pendakur.
\newblock Unobserved preference heterogeneity in demand using generalized
  random coefficients.
\newblock \emph{Journal of Political Economy}, 125\penalty0 (4):\penalty0
  1100--1148, 2017.

\bibitem[Ma and Song(2015)]{ma2015}
S.~Ma and P.~X.-K. Song.
\newblock Varying index coefficient models.
\newblock \emph{Journal of the American Statistical Association}, 110\penalty0
  (509):\penalty0 341--356, 2015.

\bibitem[Mammen(1993)]{mammen1993}
E.~Mammen.
\newblock Bootstrap and wild bootstrap for high dimensional linear models.
\newblock \emph{The Annals of Statistics}, pages 255--285, 1993.

\bibitem[Masten(2018)]{masten2014}
M.~A. Masten.
\newblock Random coefficients on endogenous variables in simultaneous equations
  models.
\newblock \emph{The Review of Economic Studies}, 85\penalty0 (2):\penalty0
  1193--1250, 2018.

\bibitem[Newey(1997)]{newey1997}
W.~K. Newey.
\newblock Convergence rates and asymptotic normality for series estimators.
\newblock \emph{Journal of Econometrics}, 79\penalty0 (1):\penalty0 147--168,
  1997.

\bibitem[Newey and Powell(2003)]{NP03econometrica}
W.~K. Newey and J.~L. Powell.
\newblock Instrumental variable estimation of nonparametric models.
\newblock \emph{Econometrica}, 71:\penalty0 1565--1578, 2003.

\bibitem[Park et~al.(2015)Park, Mammen, Lee, and Lee]{park2015varying}
B.~U. Park, E.~Mammen, Y.~K. Lee, and E.~R. Lee.
\newblock Varying coefficient regression models: a review and new developments.
\newblock \emph{International Statistical Review}, 83\penalty0 (1):\penalty0
  36--64, 2015.

\bibitem[Pollard(2002)]{2002Pollard}
D.~Pollard.
\newblock \emph{A user's guide to measure theoretic probability}, volume~8.
\newblock Cambridge University Press, 2002.

\bibitem[Schennach(2007)]{schennach2007instrumental}
S.~M. Schennach.
\newblock Instrumental variable estimation of nonlinear errors-in-variables
  models.
\newblock \emph{Econometrica}, 75\penalty0 (1):\penalty0 201--239, 2007.

\bibitem[van~der Vaart and Wellner(2000)]{Vaart2000}
A.~van~der Vaart and J.~Wellner.
\newblock \emph{{Weak Convergence and Empirical Processes: With Applications to
  Statistics (Springer Series in Statistics)}}.
\newblock Springer, 2000.

\bibitem[Xia and Li(1999)]{xia1999}
Y.~Xia and W.~Li.
\newblock On single-index coefficient regression models.
\newblock \emph{Journal of the American Statistical Association}, 94\penalty0
  (448):\penalty0 1275--1285, 1999.

\bibitem[Xue and Wang(2012)]{xue2012}
L.~Xue and Q.~Wang.
\newblock Empirical likelihood for single-index varying-coefficient models.
\newblock \emph{Bernoulli}, 18\penalty0 (3):\penalty0 836--856, 2012.

\end{thebibliography}

\end{document}